\documentclass[a4paper,12pt]{article}

\usepackage[a4paper,
textwidth = 16cm,
textheight = 22.5cm,
top = 3cm]{geometry}

\usepackage{amsmath,amsthm,amssymb,bbm,graphicx,xcolor,natbib,url}
\usepackage{float}
\usepackage[ruled, vlined, linesnumbered]{algorithm2e}
\usepackage[colorlinks,citecolor=blue,urlcolor=blue,breaklinks]{hyperref}

\DeclareMathOperator*{\argmax}{arg\,max}
\DeclareMathOperator*{\argmin}{arg\,min}
\let\leq=\leqslant   
\let\geq=\geqslant

\newcommand{\RN}[1]{%
  \textup{\uppercase\expandafter{\romannumeral#1}}%
}
\title{Estimation of high-dimensional change-points under a group sparsity structure}
\author{Hanqing Cai and Tengyao Wang\footnote{Research supported by EPSRC grant EP/T02772X/1.}\\
University College London}
\date{\today}
\setcitestyle{authoryear,round}
\newtheorem{theorem}{Theorem}
\newtheorem{prop}[theorem]{Proposition}
\newtheorem{lemma}[theorem]{Lemma}

\begin{document}
\maketitle
\begin{abstract}
Change-points are a routine feature of `big data' observed in the form of high-dimensional data streams. In many such data streams, the component series possess group structures and it is natural to assume that changes only occur in a small number of all groups. We propose a new change point procedure, called \texttt{groupInspect}, that exploits the group sparsity structure to estimate a projection direction so as to aggregate information across the component series to successfully estimate the change-point in the mean structure of the series. We prove that the estimated projection direction is minimax optimal, up to logarithmic factors, when all group sizes are of comparable order. Moreover, our theory provide strong guarantees on the rate of convergence of the change-point location estimator. Numerical studies demonstrates the competitive performance of \texttt{groupInspect} in a wide range of settings and a real data example confirms the practical usefulness of our procedure.
\end{abstract}

\section{Introduction}
Modern applications routinely generate time-ordered high-dimensional datasets, where many covariates are simultaneously measured over time. Examples include wearable technologies recording the health state of individuals from multi-sensor feedbacks \citep{HanlonAnderson2009}, internet traffic data collected by tens of thousands of routers \citep{PengLeckieRamamohanarao2004} and functional Magnetic Resonance Imaging (fMRI) scans that record the time evolution of blood oxygen level dependent (BOLD) chemical contrast in different areas of the brain \citep{AstonKirch2012}. The explosion in number of such high-dimensional data streams calls for methodological advances for their analysis.

Change-point analysis is an essential statistical technique used in identifying abrupt changes in a time series. Time points at which such abrupt change occurs are called `change-points'. Through estimating the location of change-points, we can divide the time series into shorter segments that can be analysed using methods designed for stationary time series. Moreover, in many applications, the estimated change-points indicate specific events that are themselves of great interest. In the examples mentioned in the previous paragraph, they can be used to raise alarms about abnormal health events, detect distributed denial of service attacks on the network and pinpoint the onset of certain brain activities.

Classical change-point analysis focuses on univariate time series. The current state-of-art methods including \citet{KFE2012, FMS2014, Fryzlewicz2014}. However, classical univariate change-point methods are often inadequate for high-dimensional datasets that are routinely encountered in modern applications. When applied componentwise, they are often sub-optimal as signals can spread over many components. As a result, several new methodologies have been proposed to test and estimate change-points in the high-dimensional settings. These include methods that apply a simple $\ell_2$ or $\ell_\infty$ aggregation of test statistics across different components \citep{HorvathHuskova2012, Jirak2015}, and more complex methods such as a scan-statistics based approach by \citet{EnikeevaHarchaoui2019}, the Sparsified Binary Segmentation algorithm by \citet{ChoFryzlewicz2015}, the double CUSUM algorithm of \citet{Cho2016} and a projection-based approach by \citet{WS2018}. 

To get around the issue of the curse of dimensionality, existing high-dimensional change-point methods often assume that the signal of change possesses some form of sparsity. For example, in the high-dimensional mean change setting studied in \citet{Jirak2015, ChoFryzlewicz2015, WS2018, EnikeevaHarchaoui2019}, it is assumed that the difference in mean before and after a change-point is nonzero only in a small subset of coordinates. While the sparsity assumption greatly reduces the complexity of the original high-dimensional problem, it often does not capture the the full extent of the structure in the vector of change available in real data applications. For instance, in many applications, the coordinates of the high-dimensional vectors are naturally clustered into groups and coordinates within the same group tend to change together. At each change-point, only a small number of groups will undergo a change. Such a group sparsity change-point structure is useful in modelling many practical applications.   Examples  include  financial  data  stream  where  changes  are  often  grouped  by industry  sectors  and  a  small  number  of  sectors  may  experience  virtually  simultaneous market shocks.  Also, in functional magnetic resonance imaging data, voxels belonging to the same brain functional regions tend to change simultaneously over time. Similar group sparsity assumptions have been made in other statistical problems including \citet{YuanLin2006,WL2008,simon2013}.

    In this work, we provide a new high-dimensional change-point methodology that exploits the group sparsity structure of the changes. More precisely, given pre-specified grouping information of all the coordinates, our algorithm, named \texttt{groupInspect} (standing for \textbf{group}-based \textbf{in}formative \textbf{s}parse \textbf{p}rojection \textbf{e}stimator of \textbf{c}hange-poin\textbf{t}s), will first estimate a vector of projection that is closely aligned with the true vector of change at each change-point. It will then project the high-dimensional data series along this estimated direction and apply a univariate change-point method on the projected series to identify the location of the change. The above procedure can be combined with a wild binary-segmentation algorithm \citep{Fryzlewicz2014} to recursively identify multiple change-points. We show that, in a single change-point setting, the projection direction estimator employed in \texttt{groupInspect} has a minimax optimal dependence, up to logarithmic factors, on both the $\ell_0$ sparsity parameter and the group-sparsity parameter, representing respectively the number of nonzero elements and the number of nonzero groups in the vector of change. Furthermore, \texttt{groupInspect} achieves a $\sqrt{\log(n)} / (n\vartheta^2)$ rate of convergence for the estimated location of a single change-point, where $\vartheta$ denotes the $\ell_2$ norm of the vector of change, which up to logarithmic factors is minimax optimal. 

The outline of the paper is as follows. In Section~\ref{Sec:Description}, we describe the formal setup of our problem. The \texttt{groupInspect} methodology is then introduced in Section~\ref{Sec:Method}, with its theoretical performance guarantees provided in Section~\ref{Sec:Theory}. We illustrate the empirical performance of \texttt{groupInspect} via simulatinos and a real-data example in Section~\ref{Sec:Simulations}. Proofs of all theoretical results are deferred to Section~\ref{Sec:Proofs}, and ancillary results and their proofs are given in Section~\ref{Sec:Ancillary}.

\subsection{Notation}
For any positive integer $n$, we write $[n] = \{1, \dots, n\}$. For a vector $v = (v_1, \ldots, v_n)^\top \in \mathbb{R}^n$, we define $\|v\|_0 = \sum_{i=1}^n \mathbbm{1}_{\{v_i \neq 0\}}$, $\|v\|_\infty = \max_{i \in [n]} |v_i|$ and $\|v\|_q=\bigl\{\sum_{i=1}^{n} (v_i)^q\bigr\}^{1/q}$ for any positive integer $q$, and let $\mathbb{S}^{n-1} = \{ v \in \mathbb{R}^n : \| v\|_2 = 1\}$.  For a matrix $A\in\mathbb{R}^{p\times n}$, we write $\|A\|_*$ for its nuclear norm and write $\|A\|_{\mathrm{F}}$ for its Frobenius norm. 

For any $S\subseteq [n]$, we write $v_S$ for the $|S|$-dimensional vector obtained by extracting coordinates of $v$ in $S$. For a matrix $A\in\mathbb{R}^{p\times n}$, $J\in [p]$ and $S\in[n]$, we write $A_{J,S}$ for the submatrix obtained by extracting rows and columns of $A$ indexed by $J$ and $S$ respectively. When $S = [n]$, we abbreviate $A_{J,[n]}$ by $A_{J}$. When $S=\{t\}$ is a single element set, we slightly abuse notation and write $A_{J,t}$ instead of $A_{J,\{t\}}$. 

Given two sequences $(a_n)_{n\in\mathbb{N}}$ and $(b_n)_{n\in\mathbb{N}}$ such that $a_n, b_n > 0$ for all $n$, we write $a_n \lesssim b_n$ (or equivalently $b_n \gtrsim a_n$) if $a_n \leq C b_n$ for some universal constant $C$.  

\section{Problem description}
\label{Sec:Description}
Let $X_1,\ldots, X_n$ be independent random vectors with distribution:
\begin{equation}
\label{Eq:DGM1}
X_t \sim N_p(\mu_t,\sigma^2 I_p), \quad 1\leq t\leq n,
\end{equation}
which we can combine into a single data matrix $X \in \mathbb{R}^{p\times n}$. We assume that the sequence of mean vectors $(\mu_t)_{t=1}^n$ undergoes changes at times $z_i \in\{1,\ldots,n-1\}$ for $i\in\{1,\ldots,\nu\}$, in the sense that 
\begin{equation}
\label{Eq:DGM3}
\mu_{z_i+1}=\cdots=\mu_{z_{i+1}}=:\mu^{(i)}, \quad \forall\, i \in \{0,\ldots,\nu\},
\end{equation}
where we use the convention that $z_0=0$ and $z_{\nu+1}=n$. We assume that consecutive change-points are sufficiently separated in the sense that
\begin{equation*}
    \min\{z_{i+1}-z_i: 0\leq i\leq\nu\}\geq n\tau.
\end{equation*}
Suppose further that each of the $p$ coordinates belong to (at least) one of $G$ groups. Specifically, let $\mathcal{J}_g$ denotes the set of indices associated with the $g$th group for $g\in\{1,\ldots,G\}$, we have that
\begin{equation}
\label{Eq:DGM2}
\bigcup_{g=1}^G \mathcal{J}_g = [p].
\end{equation}
We assume that coordinates in the same group will tend to change together. We will consider both the case of overlapping and non-overlapping groups. In the latter scenario, each coordinate belongs to a unique group and $(\mathcal{J}_g)_{g\in[G]}$ forms a partition of $[p]$. 

Our goal is to estimate the locations of change $z_1,\ldots,z_\nu$ from the data matrix $X$ and the pre-specified grouping information $(\mathcal{J}_g)_{g\in[G]}$. Motivated by \citet{WS2018}, the best way to aggregate the component series so as to maximise the signal-to-noise ratio around the $i$th change-point is to project the data along a direction close to the vector of change $\theta^{(i)}= \mu^{(i)} - \mu^{(i-1)}$. Let $v^{(i)}$ be the unit vector parallel to $\theta^{(i)}$: 
\begin{equation*}
\label{Eq:OracleDirection}
    v^{(i)}=\theta^{(i)}/\|\theta^{(i)}\|_2,
\end{equation*}
which we will call the oracle direction for the $i$th change-point. We measure the quality of any estimated projection direction $\hat v$ with the Davis--Kahan sin $\theta$ loss \citep{DavisKahan1970}
\[
L(\hat v, v^{(i)}) = \sqrt{1 - (\hat v^\top v^{(i)})^2}
\]
and measure the quality of the subsequent location estimator $\hat z_i$ by $\mathbb{E}|\hat z_i - z_i|$.

The difficulty of the estimation task depends on both the noise level $\sigma$ and the vector of change $\theta^{(i)}= \mu^{(i)} - \mu^{(i-1)}$.  More precisely, we assume that the change is localised in a small number of the $G$ groups as defined in \eqref{Eq:DGM2}. Define $\phi: \mathbb{R}^p\to \mathbb{R}^G$ such that $\phi(x)=(\|x_{\mathcal{J}_1}\|_2, \|x_{\mathcal{J}_2}\|_2, \ldots, \|x_{\mathcal{J}_G}\|_2)^\top $, we assume that
\begin{equation}
\label{Eq:DGM4}
\|\theta^{(i)}\|_0\leq k, \quad \|\phi(\theta^{(i)})\|_0\leq s \quad \text{and} \quad \|\theta^{(i)}\|_2 \geq \vartheta.
\end{equation}

\section{Methodology}
\label{Sec:Method}
\subsection{Single change-point estimation}
\label{Sec:SingleCP}
Initially, we will consider estimation of a single change-point, where $\nu=1$. This can be extended to estimate multiple change-points in conjunction with top-down approaches such as wild binary segmentation, which we will discuss in Section~\ref{Sec:MultipleCP}. 

We define the CUSUM transformation $\mathcal{T}: \mathbb{R}^{p\times n}\to \mathbb{R}^{p\times (n-1)}$ by
\begin{equation}
\label{eq:cusum}
  \mathcal{T}(M)_{j,t}=\sqrt{\frac{t(n-t)}{n}}\bigg(\frac{1}{n-t}\sum_{r=t+1}^nM_{j,r}-\sum_{r=1}^t\frac{1}{t}M_{j,r}\bigg), 
\end{equation}
and compute the CUSUM matrix $T=\mathcal{T}(X)$. As discussed in Section~\ref{Sec:Description}, our general strategy is to use the matrix $T$ to estimate a projection direction that is well-aligned with the direction of change, and then project the data along this direction to estimate the change-point location from the univariated projected series. More precisely, we would like to solve for
\begin{equation}
\label{Eq:NonConvex}
\hat v \in \argmax_{u\in\mathbb{S}^{p-1}, \|\phi(u)\|_0\leq s}  \|u^\top T\|_2,
\end{equation}
where, $\mathbb{S}^{p-1}=\{x\in \mathbbm{R}^p: \|x\|_2=1\}$.
However, the above optimisation problem is non-convex due to the group-sparsity constraint. Consequently, we perform the following convex relaxation of the above problem. We first note that the set of optimisers of~\eqref{Eq:NonConvex} is equal to the set of leading left singular vectors of 
\[
\argmax_{\substack{M\in\mathbb{R}^{p\times (n-1)}:\|M\|_*=1, \mathrm{rank}(M)=1\\ \sum_{g\in[G]}\mathbbm{1}_{\{\|M_{\mathcal{J}_g}\|_{\mathrm{F}}\neq 0\}\leq s}}} \langle M,T\rangle,
\]
We relax the above matrix-variate optimisation problem by dropping the combinatorial rank constraint, and replacing the nuclear norm constraint set by the larger Frobenius norm set of $\mathcal{S}=\{M\in\mathbb{R}^{p\times (n-1)}: \|M\|_F\leq 1\}$.  The constraint that $M$ has at most $s$ groups of non-zero rows can be written as an $\ell_0$ constraint on the vector of Frobenius norms of such submatrices, i.e.\ $\|(\|M_{\mathcal{J}_g}\|_F:g\in\{1,\ldots,G\})\|_0 \leq s$.
Motivated by the group lasso penalty \citep{YuanLin2006}, we replace this group sparsity constraint with a \emph{group norm} penalty, where the group norm for a matrix $M\in\mathbb{R}^{p\times (n-1)}$ is defined as
\[
\|M\|_\mathrm{grp}=\sum_{g=1}^{G} p_g^{1/2}\|M_{\mathcal{J}_g}\|_{2,1},
\]
where $\|M_{\mathcal{J}_g}\|_{2,1}$ is the sum of column $\ell_2$ norms of the submatrix $M_{\mathcal{J}_g}$ and $p_g=|\mathcal{J}_g|$. Overall, we obtain the following optimisation problem:
\begin{equation}
\label{Eq:GrpSparsityOptimization}
\hat M \in \argmax_{M\in \mathcal{S}}\big\{\langle T,M\rangle- \lambda \|M\|_{\mathrm{grp}}\bigr\},
\end{equation}
where $\lambda\in[0,\infty)$ is a regularization parameter. 

If the groups are non-overlapping, in the sense that $\mathcal{J}_g\cap \mathcal{J}_{g'}=\emptyset$ for all $g\neq g'$, then we see from Proposition~\ref{prop: closed form} that~\eqref{Eq:GrpSparsityOptimization} has a closed form solution
\begin{equation}
\hat M=\frac{T-R^*}{\left\lVert T-R^*\right\rVert_\mathrm{F}},
\label{Eq:ClosedForm}
\end{equation}
where $R_{\mathcal{J}_g,t}^*=T_{\mathcal{J}_g,t} \min\big\{\frac{\lambda p_g^{1/2}}{\|T_{\mathcal{J}_g,t}\|_2}, 1\big\}$. 

For overlapping groups, \eqref{Eq:GrpSparsityOptimization} can be optimised using Frank--Wolfe algorithm \citep{FK1956}, as described in Algorithm~\ref{alg:frank wolfe}. We first compute the gradient of the objective function which is the step 4 in Algorithm~\ref{alg:frank wolfe}. We then project the $\hat{M}$ back onto $\mathcal{S}$. 

After solving the optimization problem, we can obtain the estimated projection direction $\hat{v}$ by computing the leading left singular vector of $\hat{M}$. Then, we project the data along $\hat{v}$ to obtain a univariate series for which existing one-dimensional change-point estimation methods apply. Specifically, we perform the CUSUM transformation over the projected data series, and locate the change-point by the maximum absolute value of the CUSUM vector. The full procedure is described in Algorithm~\ref{alg:single change}.

\begin{algorithm}[htbp]
 \caption{Frank--Wolfe algorithm for optimising~\eqref{Eq:GrpSparsityOptimization}}
  \label{alg:frank wolfe}
    \KwIn{ $T\in\mathbb{R}^{p\times (n-1)}$, grouping $(\mathcal{J}_g)_{g\in[G]}$, $\lambda>0$ and $\epsilon>0$.}
    Initialise $\hat M^{[0]}= T / \|T\|_{\mathrm{F}}$ and $t=0$.\\
    \Repeat{$\|\hat M^{[t+1]} - \hat M^{[t]}\|_{\mathrm{F}} \leq \epsilon$}{
    $t\leftarrow t+1$\\
      Compute $G^{[t]} = (G_1^{[t]},\ldots,G_p^{[t]})^\top \in \mathbb{R}^{p\times (n-1)}$ such that
      \[
    G_{j,t}^{[t]} \leftarrow T_{j,t}-\sum_{g:j\in \mathcal{J}_g}\lambda_g\frac{M_{j,t}^{[t-1]}}{\|M_{\mathcal{J}_{g,t}}^{[t-1]}\|_{\mathrm{F}}}, 
    \] where $\lambda_g=p_g^{1/2}\lambda$\\
    \textbf{if $G^{[t]} = 0$ then break}\\
    Compute
    \[
    \tilde{M}^{[t]}=\frac{t}{t+2} M^{[t-1]}+\frac{2}{t+2} \frac{G^{[t]}}{\|G^{[t]}\|_{\mathrm{F}}},
    \]\\
    Normalise $\hat M^{[t]} \leftarrow \tilde{M}^{[t]} / \|\tilde{M}^{[t]}\|_{\mathrm{F}}$
    }
   \KwOut{$\hat M^{[t]}$}
\end{algorithm}

\begin{algorithm}[htbp]
 \caption{Single change-point estimation procedure for data with group structure}
  \label{alg:single change}
    \KwIn{$X\in \mathbb{R}^{p\times n}$, $(\mathcal{J}_g)_{g\in[G]}$, and $\lambda > 0$}
    Compute $T \leftarrow \mathcal{T}(X)$ as in~\eqref{eq:cusum}.\\
    Solve 
    \[
    \hat M \in \argmax_{M\in \mathcal{S}} \bigl\{\langle T, M\rangle - \lambda \|M\|_{\mathrm{grp}}\bigr\}
    \]
    using either the closed-form solution in~\eqref{Eq:ClosedForm} if groups are non-overlapping, or Algorithm~\ref{alg:frank wolfe}.\\
    Let $\hat v$ be the leading left singular vector of $\hat M$.\\
    Estimate $z$ by $\hat{z}=\argmax_{1\leq t\leq n-1}|\hat{v}^\top T_t|$, where $T_t$ is the $t$th column of $T$.\\
    \KwOut{$\hat{z}$, $\bar{T}_{\max}=\hat v^\top T_z$}
\end{algorithm}

\subsection{Multiple change-point estimation}
\label{Sec:MultipleCP}
When the data matrix possess multiple change-points, we may combine Algorithm~\ref{alg:single change} with a top-down approach, such as the wild binary segmentation \citep{Fryzlewicz2014}, to recursively identify all the change-points. Specifically, we start by drawing a large number of random intervals $[s_1,e_1],\ldots,[s_Q,e_Q]$ and apply Algorithm~\ref{alg:single change} to the data matrix $X$ restricted to each of these time intervals to obtain $Q$ candidate change-point locations. We then aggregate $Q$ candidate change-point locations to choose the one with the maximum projected CUSUM statistics. If the value of the CUSUM statistic at the best candidate location is above a threshold $\xi$, we will admit this candidate location as a change-point and repeat the above process on the data submatrix to the left and right of this change-point. The pseudocode for the full procedure is given in Algorithm~\ref{alg:multiple change}.

\begin{algorithm}[htbp]
 \caption{Multiple change-point estimation procedure}
  \label{alg:multiple change}
  \SetKwProg{Fn}{Function}{}{end}

    \KwIn{$X\in \mathbb{R}^{p\times n}$, $(\mathcal{J}_g)_{g\in[G]}$, $\lambda > 0$, $\xi>0$, $Q\in\mathbb{N}$}
    Set $\hat{Z}\leftarrow \emptyset$\\
    Draw $Q$ pairs of integers $(s_1,e_1),\ldots,(s_Q,e_Q)$ uniformly at random from the set $\{(\ell,r)\in\mathbb{Z}^2:0\leq \ell< r \leq n\}$\\
     \Fn{\textbf{\textup{wbs($s$, $e$)}}}{
        Set $Q_{s,e}=\{q:s\leq s_q<e_q\leq e \}$\\
        \For {$q \in Q_{s,e}$}{
        $(\hat{z}^{[q]}, \bar{T}^{[q]}_{\max}) \leftarrow $ output of Algorithm~\ref{alg:single change} with inputs $(X_{j,t})_{j\in[p],t\in(s,e]}$ and $\lambda$}
        Find $q_0\in\argmax_{q\in Q_{s,e}}\bar{T}^{[q]}_{\max}$ and set $b\leftarrow s_{q_0}+\hat{z}^{[q_0]}$\\
        \If {$\bar{T}^{[q_0]}_{\max}\geq \xi$}{$\hat{Z}\leftarrow \hat{Z}\cup \{b\}$\\
        Run recursively $\textbf{wbs}(s,b)$ and $\textbf{wbs}(b,e)$}
    }
    \KwOut{$\hat Z$}
\end{algorithm}

\section{Theoretical guarantees}
\label{Sec:Theory}
In this section, we provide theoretical guarantees to the performance of the groupInspect algorithm. As we have noted in Section~\ref{Sec:Description}, a key to the successful change-point estimation in the current problem is a good estimator of the oracle projection direction $v=\theta/\|\theta\|_2$. 

The following theorem controls the sine angle risk of the estimated projection direction $\hat v$ in Step 3 of Algorithm~\ref{alg:single change}. We define $\mathcal{P}_{n,p}(s, k, \tau, \vartheta,\sigma^2, (\mathcal{J}_g)_{g\in[G]})$ to be the set of data distributions satisfying~\eqref{Eq:DGM1}, \eqref{Eq:DGM3}, \eqref{Eq:DGM2} and \eqref{Eq:DGM4}. For any $P\in\mathcal{P}$, we write $v(P) = \theta/\|\theta\|_2$ where $\theta$ is the difference between post-change and pre-change means.

\begin{theorem}
\label{thm:upper bound}
For a given grouping $(\mathcal{J}_g)_{g\in[G]}$, let $p_*=\min_{g\in[G]}|\mathcal{J}_g|$ and suppose further that there exists a universal constant $C_1>0$, such that $\max_{j\in[p]} |\{g:j\in \mathcal{J}_g\}| \leq C_1$.  Let $X \sim P \in \mathcal{P}_{n,p}(s,k, \tau,\vartheta,\sigma^2, (\mathcal{J}_g)_{g\in[G]})$ be a $p\times n$ data matrix,  let $\theta$ be the vector of change and let $\hat{v}$ be as in Step~3 of Algorithm~\ref{alg:single change} with input $X$, $(\mathcal{J}_g)_{g\in[G]}$ and $\lambda\geq\sigma( 1+\sqrt{4\log(Gn)/p_*})$. Then there exists $C>0$, depending only on $C_1$, such that
\begin{equation}
\label{Eq:TwoTermLoss}
\sup_{P\in\mathcal{P}_{n,p}(s,k,\tau,\vartheta,\sigma^2, (\mathcal{J}_g)_{g\in[G]})} \mathbb{P}_{P}\bigg\{\sin\angle{(\hat{v}, v)} \leq \frac{C\lambda k^{1/2}}{n^{1/2}\tau\vartheta}
\bigg\}\leq \frac{1}{nG}.
\end{equation}
\end{theorem}
We remark that the condition $\max_{j\in[p]} |\{g:j\in \mathcal{J}_g\}| \leq C_1$ is to control the extent of overlapping between different groups. Specifically, it requires that each coordinate can belong to at most $C_1$ groups. In the special case when all groups $\mathcal{J}_g$ are disjoint, which is often true in practical applications, then it suffices to take $C_1 = 1$. 

 We note that, when $\lambda = \sigma (1+\sqrt{4\log(Gn)/p_*})$, with high probability, the sine angle loss in~\eqref{Eq:TwoTermLoss} has an upper bound that is proportional to $\sigma k^{1/2}n^{-1/2}\tau^{-1}\vartheta^{-1}$, similar to what has been previously observed in \citet[Proposition~1]{WS2018}. However, Theorem~\ref{thm:upper bound} reveals an interesting interaction between the $\ell_0$ sparsity  $k$ and the group sparsity $s
$ when all groups are of comparable size. Specifically, for $\lambda=  \sigma (1+\sqrt{4\log(Gn)/p_*})$ and assuming that $\max_{g \in [G]} p_g \lesssim p_*$, then we can simplify~\eqref{Eq:TwoTermLoss} to obtain that
\[
\mathbb{E}\{\sin\angle{(\hat{v}, v)}\} \lesssim \sqrt\frac{\sigma^2\{k  + s\log(Gn)\}}{n\tau^2\vartheta^2}.
\]
In other words, the risk upper bound undergoes a phase transition as the number of coordinates per group increases above a $\log(Gn)$ level. Similar phase transitions have been previously observed in the context of high-dimensional linear model where the regression coefficients satisfy a group sparsity assumption \citep[see, e.g.][Theorem~3]{Caietal2019}. 

We now turn our attention to a minimax lower bound of the estimation risk of the oracle projection direction. Theorem~\ref{thm:lower bound} below shows that the phase transition observed in Theorem~\ref{thm:upper bound} is not due to the specific proof techniques employed but rather an intrinsic feature of the problem. 

\begin{theorem}
\label{thm:lower bound}
Suppose $s > 0$, $k>0$ and a grouping $(\mathcal{J}_{g})_{g\in[G]}$ satisfy that $\mathcal{J}_g\cap\mathcal{J}_{g'}=\emptyset$ for all $g\neq g'$ , $\min\{k, (s-1)\log(G/s)\}\geq 20$, and  $\sum_{r=1}^{s} p_{(G-r+1)} \geq k/2$, where $p_{(1)}\leq p_{(2)}\leq \cdots\leq p_{(G)}$ are order statistics of $p_1,\ldots,p_G$. Then for some universal constant $c>0$, we have 
\[
\inf_{\tilde v} \sup_{P\in\mathcal{P}_{n,p}(s,k,\tau,\vartheta,\sigma^2, (\mathcal{J}_g)_{g\in[G]})} \mathbb{E}_{P} L(\tilde v(X), v(P)) \geq c  \sqrt\frac{\sigma^2\{k + s\log (G/s)\}}{n\tau\vartheta^2},
\]
where the infimum is taken over the set of all measurable functions $\tilde v$ of the data $X$. 
\end{theorem}

The condition that $\sum_{r=1}^s p_{(G-r+1)} \geq k/2$ is to ensure that the upper bound $k$ on the $\ell_0$-sparsity is not too loose in the sense that $k$ is not too much larger than the cardinality of the union of the largest $s$ groups. If we assume that $\log(G/s)  \asymp \log(n)$, $\tau\asymp 1$ and $\max_{g\in[G]}p_g \lesssim p_*$, then the lower bound in Theorem~\ref{thm:lower bound} matches the upper bound of Theorem~\ref{thm:upper bound} up to universal constants, when all groups are non-overlapping.


After obtaining guarantees on the quality of the projection direction estimator, we now provide theoretical guarantees of the overall change-point procedure. We note that the projection direction estimator $\hat{v}$ is dependent on the CUSUM panel $T$. While this dependence is observed to be very weak in practice, it creates difficulties in analysing the projected CUSUM series $\hat{v}^\top T$ in Step 4 of Algorithm~\ref{alg:single change}. As such, for theoretical convenience, we will instead analyse a sample-splitting version of the algorithm. Specifically, we split the data into $X^{(1)}$ and $X^{(2)}$, consisting of odd and even time points respectively, as described in Algorithm~\ref{alg:split}. We use $X^{(1)}$ to estimate the projected direction $\hat{v}^{(1)}$ and then project $X^{(2)}$ along this direction to locate the change-point. Theorem~\ref{Thm:change-pointEstimate} below provides a performance guarantee for the estimated location of the change-point of this sample-splitting version of our procedure.
\begin{algorithm}[htbp]
 \caption{Change-point estimation procedure: sample splitting version }
 \label{alg:split}
    \KwIn{$X\in \mathbb{R}^{p\times n}$ and $\lambda > 0$}
    Define $X^{(1)}$ as $X^{(1)}_{j,t}=X_{j,2t-1}$ and $X^{(2)}$ as $X^{(2)}_{j,t}=X_{j,2t}$.\\
    Compute $T^{(1)} \leftarrow \mathcal{T}(X^{(1)})$ and $T^{(2)} \leftarrow \mathcal{T}(X^{(2)})$ as in~\eqref{eq:cusum}.\\
    Solve 
    \[
    \hat M^{(1)} \in \argmax_{M\in \mathcal{S}} \bigl\{\langle T^{(1)}, M\rangle - \lambda \|M\|_{\mathrm{grp}}\bigr\}
    \]
    using either the closed-form solution in~\eqref{Eq:ClosedForm} if groups are non-overlapping, or Algorithm~\ref{alg:frank wolfe}.\\
    Let $\hat v$ be the leading left singular vector of $\hat M^{(1)}$.\\
    Estimate $z$ by $\hat{z}=2\argmax_{1\leq t\leq n_1-1}|(\hat{v}^{(1)})^\top T_t^{(2)}|$, where $T_t^{(2)}$ is the $t$th column of $T^{(2)}$.\\
    \KwOut{$\hat{z}$}
\end{algorithm}

\begin{theorem}
\label{Thm:change-pointEstimate}
Given data matrix  $X \sim P \in \mathcal{P}_{n,p}(s, k, \tau, \vartheta, \sigma^2, (\mathcal{J}_g)_{g\in[G]})$, let $\hat{z} $ be the output from the Algorithm~\ref{alg:split} with input $X$ and $\lambda=\sigma(1+\sqrt{p_*^{-1}4\log(nG)})$. There exist universal constants $C$, $C'>0$ such that, if $n\geq 12$ is even, $z$ is even, and
\begin{equation*}
\label{eq:cond 1}
    \frac{C\sigma\sqrt{k}}{\vartheta\tau\sqrt{n}}\bigg(\frac{1+\sqrt{4\log(Gn)}}{p_*}\bigg)\leq 1,
\end{equation*}
then,
\begin{equation*}
\label{eq:cond 2}
\mathbb{P}\bigg\{\frac{1}{n}|\hat{z}-z|\leq \frac{C'\sigma^2(1+\sqrt{4\log(n)})}{n\vartheta^2}\bigg\}\geq 1-\frac{20\log n}{n}.
\end{equation*}
\end{theorem}

\section{Numerical studies}
\label{Sec:Simulations}
In this section, we provide some simulation results to demonstrate the empirical performance of the groupInspect method. In all our numerical studies, unless otherwise specified, we will assume that data are generated according to~\eqref{Eq:DGM1},~\eqref{Eq:DGM3},~\eqref{Eq:DGM2} and~\eqref{Eq:DGM4}, with $\sigma = 1$. In all simulations, we do not assume that $\sigma$ is known, or even equal across rows. Instead, we estimate the variance in each row using the mean absolute deviation of successive differences of the observations. We then standardise the data by the estimated row standard deviation. The \texttt{groupInspect} procedure is then applied to the standardised data with $\sigma = 1$. 
 
\subsection{Theory validation}
We first show that the practical performance of the \texttt{groupInspect} procedure is well captured by the theoretical results in Theorems~\ref{thm:upper bound} and~\ref{thm:lower bound}. There are two related measures of the signal sparsity in our problem, which are the total number of coordinates of change $k$ and the total number of groups with a change $s$. We conduct two sets of simulation experiments fixing one of these sparsity measures and varying the other. Specifically, for $n=1000$, $p\in\{600,1200,2400\}$ and $\vartheta\in\{1,2,4,8,16\}$, we split the $p$ coordinates into disjoint groups of $p_*$ coordinates per group, where $p_*$ is allowed to vary over all divisors of $60$. In the first set of experiments, we fix $k=60$ so that $s=k/p_*$ varies with $p_*$, whereas in the second set of experiments, we fix $s=3$ so that $k=sp_*$ varies with $p_*$. The vector of change is constructed so that the magnitude of change is equal across all coordinates of change. We will use the theoretical choice of tuning parameter $\lambda$ for both sets of experiments here. Figure~\ref{fig:theory} shows how the $\sin \theta$ loss, averaged over $100$ Monte Carlo repetitions, varies with $p_*$, for different choices of $p$ and $\vartheta$ in both settings. 

In the left panel of Figure~\ref{fig:theory}, where the number of signal coordinates $k$ is fixed, we see that the average loss decreases as $p_*$ increases. Furthermore, at a log-log scale, and for relatively large signal sizes of $\vartheta \in \{4,8, 16\}$, we see the loss curves follow an initial linear decreasing trend as $p_*$ increases before plateauing eventually. This is in agreement with the two terms contributing to the loss described in Theorem~\ref{thm:upper bound}. Specifically, for small $p_*$, we expect the second term of~\eqref{Eq:TwoTermLoss} to dominate and the loss decreases at a rate approximately proportional to $1/\sqrt{p_*}$ initially. For large $p_*$, we expect the first term of~\eqref{Eq:TwoTermLoss} to dominate and the loss will have minimal dependence on $p_*$. In the right panel of Figure~\ref{fig:theory}, where the number of signal groups $s$ is fixed, the average loss increases with $p_*$, as expected from our theory. It appears that for $s=3$ studied here, the first term of~\eqref{Eq:TwoTermLoss} is dominant and the average loss increases linearly at the log-log scale with respect to $p_*$. 

We further remark that in both panels of Figure~\ref{fig:theory}, the average loss for large $p_*$ shows equally spaced separation for the signal size $\vartheta$ in the dyadic grid $\{1,2,4,8,16\}$. This is in good agreement with the $1/\theta$ dependence of expected loss given in Theorem~\ref{thm:upper bound}. Finally, we note that the ambient dimension $p$ has minimal effect on the loss curves, for all signal strengths studied here. Again, this is predicted by our theory as the dimension $p$ enters the mean loss in~\eqref{Eq:TwoTermLoss} only through the $\log(Gn) = \log(pn/p_*)$ expression in the second term.

\begin{figure}[htbp]
\centering
\includegraphics[width=0.9\textwidth]{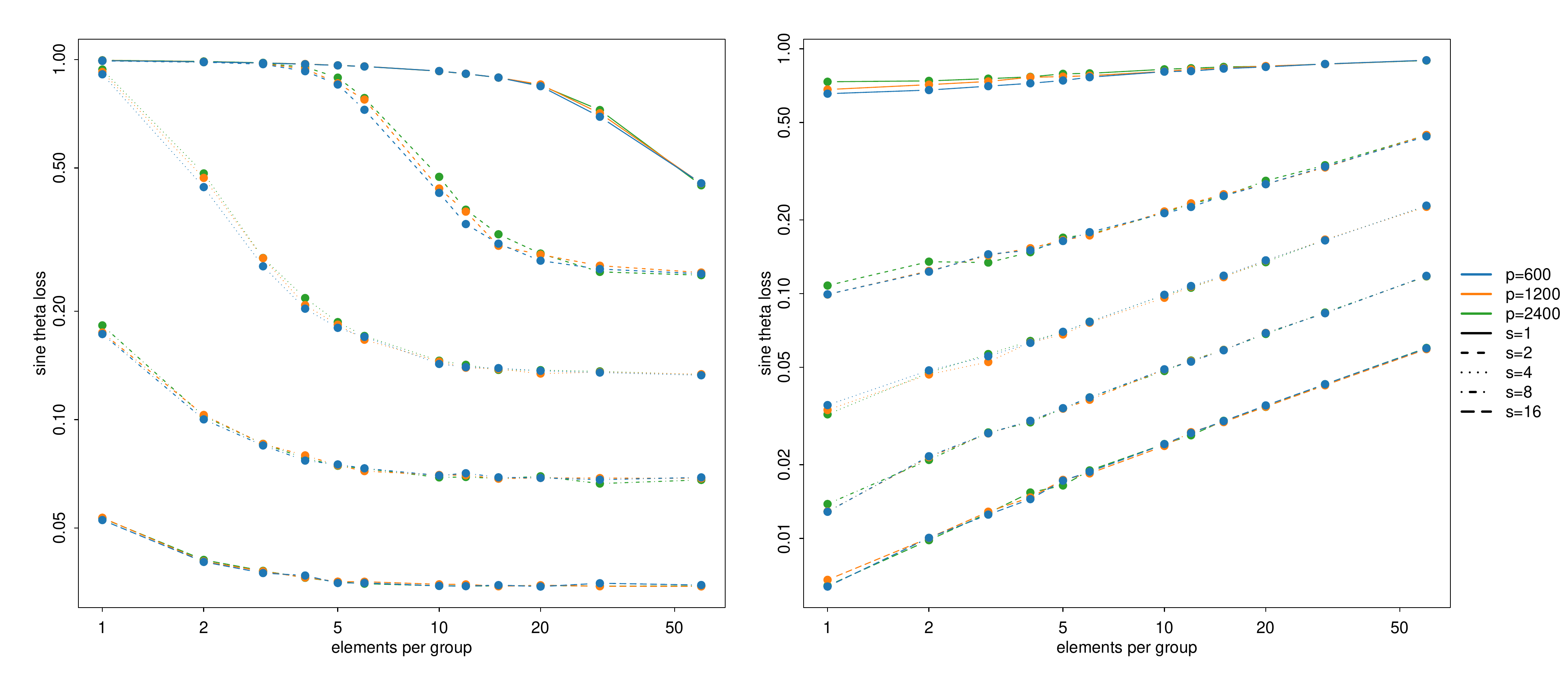} 
\caption{\label{fig:theory}Average loss (over 100 repetitions) of \texttt{groupInspect} for varying elements per group $p_*$, plotted on a log-log scale. Left panel: $k=60$ and $s=k/p_*$. Right panel: $s=3$ and $k=sp_*$. Other parameter: $n=1000$.}
\end{figure}
 
\subsection{Practical choice of tuning parameter}
\label{Sec:Tuning}
The theoretical choice of $\lambda$ turns out to be conservative in practical use. In this subsection, we will perform numerical simulations to suggest a suitable practical tuning parameter choice. We fix $n=1000$, $z=400$, $s=3$, $G\in\{10, 25\}$. The signal size $\vartheta$ is varied in $\{1,2,4,8,16\}$ and $p$ is chosen from $\{500,1000\}$. All groups are set to have equal size. For the choice of tuning parameters, we first form a logarithmic sequence of values between $0.1$ and $3$ with length $7$ and then times each value with the theoretically suggested value of $1+\sqrt{4p_*^{-1}\log(nG)}$ to form the sequence of the tuning parameter. For each setting of the signal size, we will run algorithm with all the $\lambda$ values and record the sine angle loss.

We plot $\sin \theta$ loss against $\lambda$ in Figure~\ref{fig:tun}. The $x$-axis is the log sequence. In most cases, the loss is minimized when tuning parameter value is half of the theoretical value. However, when the minimum loss is achieved by other values of $\lambda$, this lambda value can still achieve the loss which is close to optimal value. Therefore, we suggest that using $\lambda=1/2(1+\sqrt{4p_*^{-1}\log (nG)})$ in practical is less conservative.

 \begin{figure}[htbp]
\centering
\includegraphics[width=0.9\textwidth]{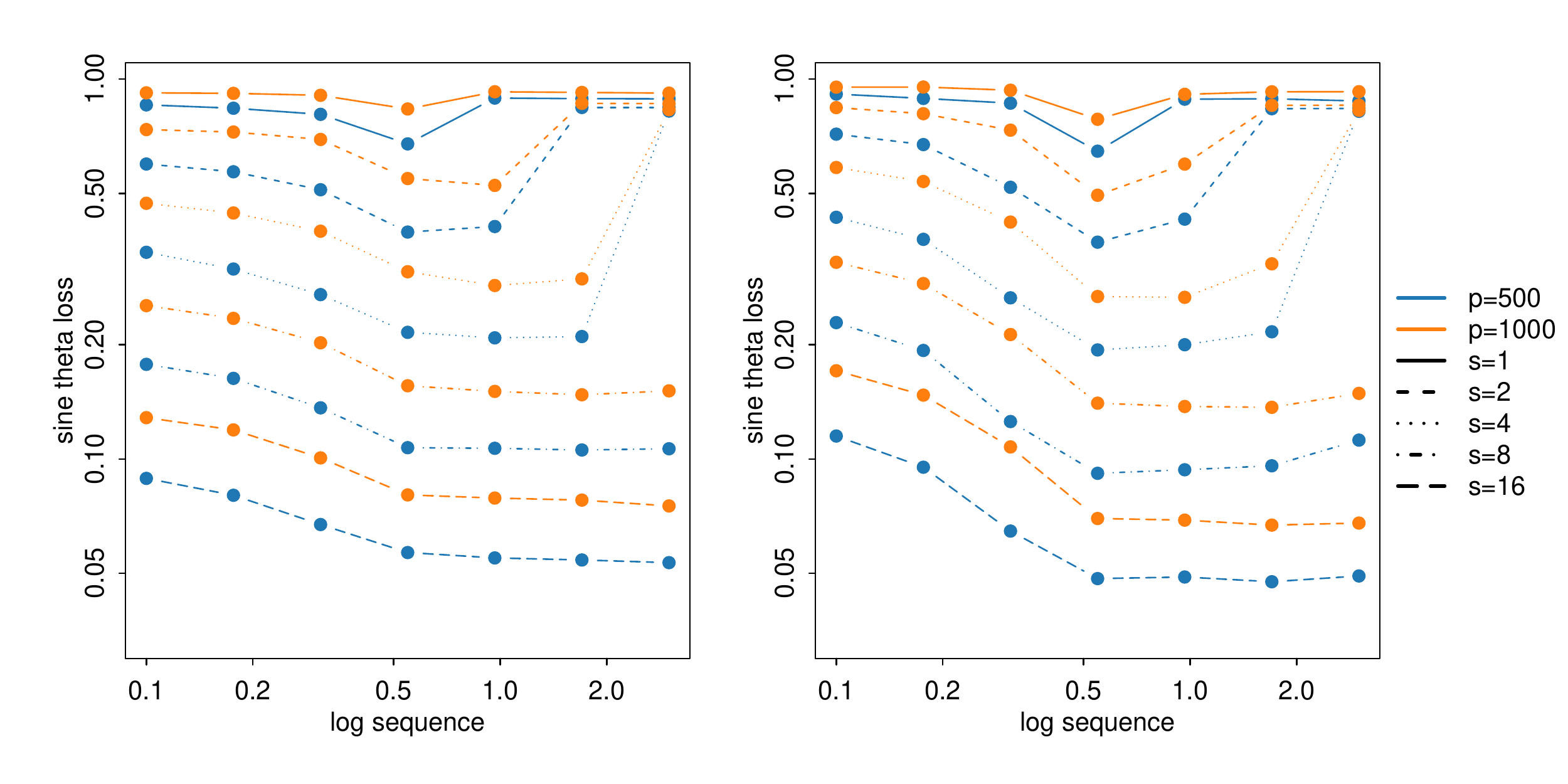} 
\caption{\label{fig:tun}Average loss (over 100 repetitions) of \texttt{groupInspect} for varying tuning parameter $\lambda$. Left panel: $G=10$. Right panel: $G=25$. Other parameter: $n=1000$, $s=3$.}
\end{figure}

\subsection{Comparison between different methods }

Now, we would like to compare our method with other existing change-point estimation procedures. As \texttt{groupInspect} is a two-stage procedure that first estimates a projection direction before localising the change-point on the projected series, we will investigate its performance both in terms of its accuracy in estimating the projection direction and the quality of the final change-point location estimator. For the former, we compare the estimated projection direction from \texttt{groupInspect} with that from the \texttt{inspect} algorithm. We measure the accuracy in terms of the sine angle loss introduced in Section~\ref{Sec:Description}. We use the recommended values for tuning parameters in both methods, i.e., $\sqrt{2^{-1}\log\{p\log n\}}$ in \texttt{inspect} as in \citet{WS2018} and $1/2(1+\sqrt{4p_*^{-1}\log (nG)})$ for \texttt{groupInspect} as suggested in Section~\ref{Sec:Tuning}. 

We fix $n=1000$, $p=1000$ and vary $\vartheta$ in $\{1,2,4,8,16\}$. We consider settings with both non-overlapping groups and overlapping groups. For the non-overlapping setting, we have $G=10$ groups of equal size $p_*=100$, whereas for the overlapping setting, we have $G=19$ groups of size $100$ each, where neighbouring groups overlap in exactly $50$ coordinates. Both methods have access to exactly the same data sets and the performance is averaged over 100 Monte Carlo repetitions.

Figure~\ref{fig:loss comparison} shows the comparison of the average sine angle loss between \texttt{groupInspect} and \texttt{inspect} over all signal sizes on a logarithmic scale, in both the non-overlapping and overlapping settings. In both cases, \texttt{groupInspect} outperforms the \texttt{inspect}  algorithm. From the left panel, we can see that the estimation accuracy of the projection direction using \texttt{groupInspect} is substantially better even when the signal is small.

We now turn our attention to the overall change-point localisation accuracy of the \texttt{groupInspect} procedure. To this end, we compare the mean absolute deviation of various high-dimensional change-point procedures over 100 Monte Carlo repetitions using the same data sets. In addition to \texttt{inspect}, we also compare against the $\ell_2$ aggregation procedures of \citet{HorvathHuskova2012}, the $\ell_\infty$ aggregation procedure of \citet{Jirak2015} and the double CUSUM procedure of \citet{Cho2016}. We set $n=1000$, $p\in\{500,1000,2000\}$, $\vartheta \in\{0.25, 0.5, 1, 2, 4\}$. The simulation results are presented in Table~\ref{table:compare}. For simplicity, we have only shown the results for 10 equal-sized non-overlapping groups here, but qualitatively similar results were obtained in other settings as well. We see that \texttt{groupInspect} is very competitive over a wide range of dimensions and signal-to-noise ratio settings, though the benefit of using the group sparsity structure via \texttt{groupInspect} is most apparent in low signal-to-noise ratio settings where the change-point estimation problem is more difficult.

\begin{figure}[htbp]
\centering
\begin{tabular}{cc}
\includegraphics[width=5cm]{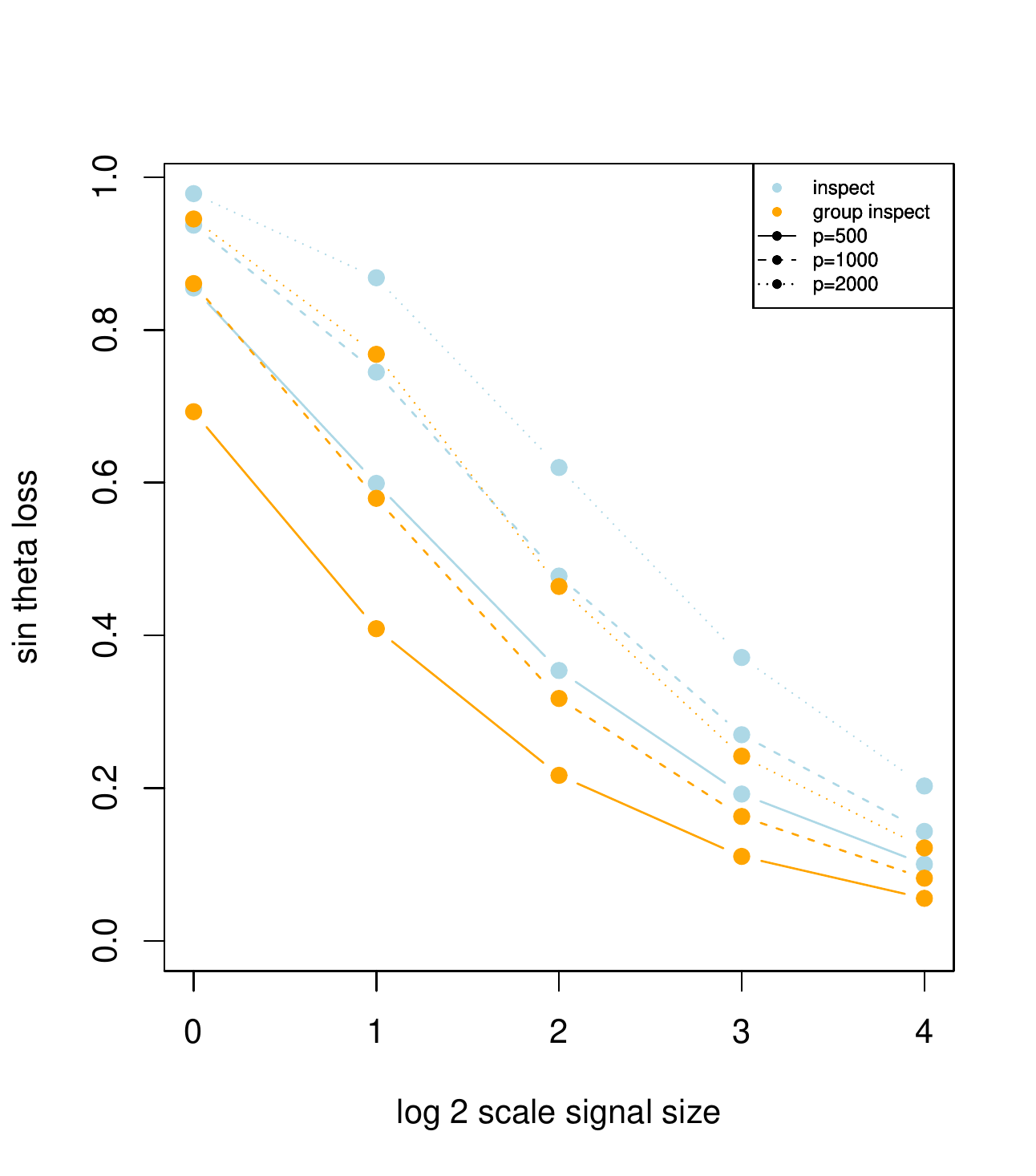} &
\includegraphics[width=5cm]{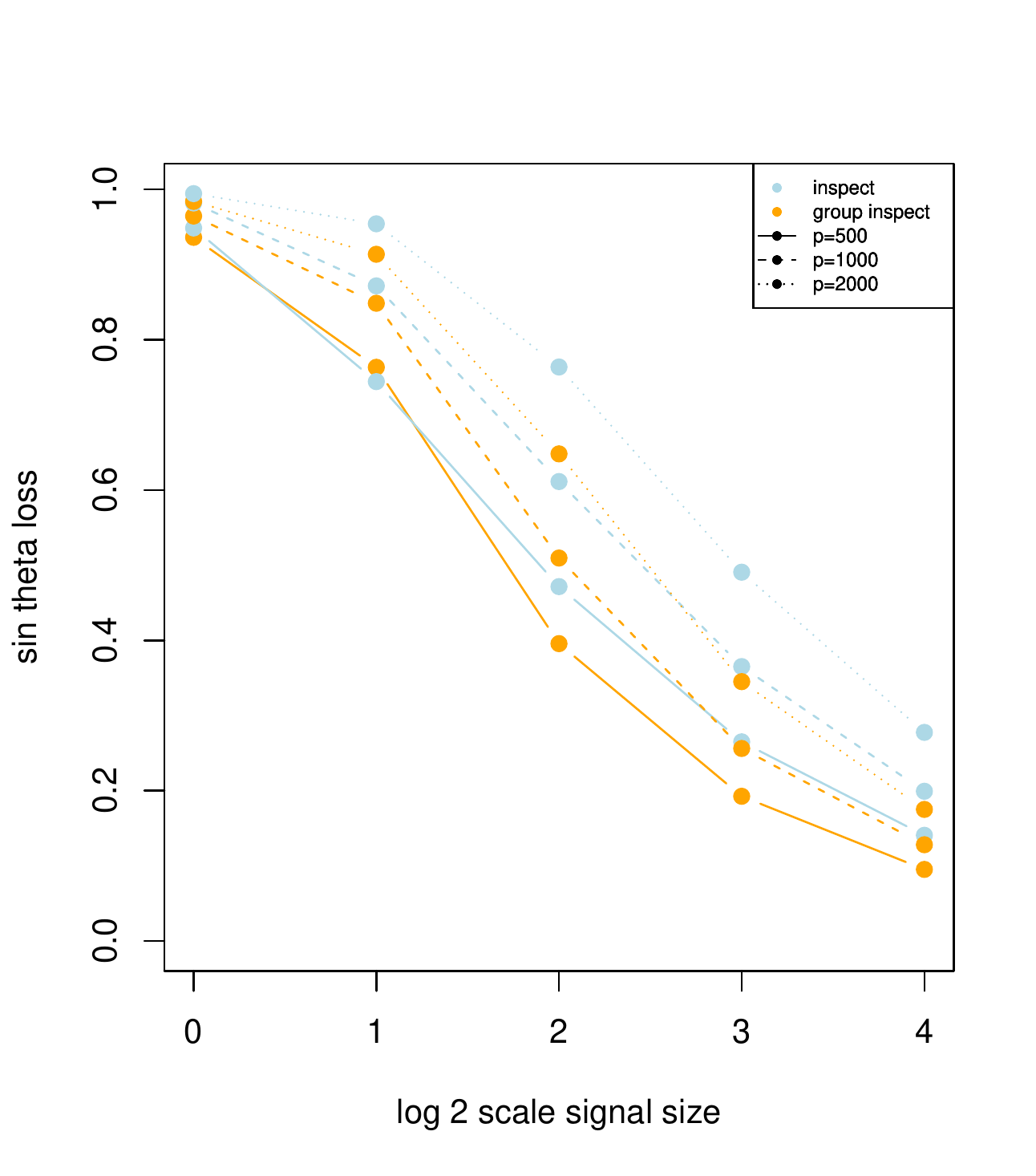}
\end{tabular}
\caption{\label{fig:loss comparison}Average loss (over 100 repetitions) comparison between \texttt{groupInspect} and \texttt{Inspect}. Left panel: non-overlap setting. Right panel: overlap setting}
\end{figure}

\begin{table}[htbp]
\begin{center}
\begin{tabular}{ccccccccc}
\hline\hline
$p$ & $\vartheta$  & \texttt{groupInspect} & \texttt{inspect} & $\ell_2$-aggregate & $\ell_\infty$-aggregate & double cusum\\
\hline
 $500$ & $0.25$ &  $\mathbf{127}$ & $143$ & $336$ & $337$ & $347$\\
$500$ & $0.5$ & $\mathbf{59.8}$ & $93.4$ & $231$ & $305$ & $262$\\
$500$ & $1$ &  $\mathbf{3.83}$ & $8.83$ & $9.84$ & $94.2$ & $40.6$\\
$500$ & $2$  & $\mathbf{0.670}$ & $0.982$ & $0.875$ & $16.0$ & $4.16$\\
$500$ & $4$ & $0.045$ & $\mathbf{0.018}$ & $0.045$ & $4.04$ & $0.179$\\
$1000$ & $0.25$ & $\mathbf{108}$ & $138$ & $347$ & $348$ & $363$\\
 $1000$ & $0.5$  & $\mathbf{81.8}$ & $107$ & $269$ & $326$ & $297$\\
 $1000$ & $1$ & $\mathbf{15.6}$ & $34.6$ & $22.1$ & $204$ & $57.9$\\
 $1000$ & $2$  & $\mathbf{0.920}$ & $1.51$ & $0.973$ & $28.3$ & $3.91$\\
$1000$ & $4$ & $\mathbf{0.081}$ & $0.117$ & $0.099$ & $6.70$ & $0.387$\\
 $2000$ & $0.25$ & $\mathbf{101}$ & $139$ & $358$ & $365$ & $364$\\
 $2000$ & $0.5$  & $\mathbf{91.2}$ & $127$ & $305$ & $353$ & $321$\\
 $2000$ & $1$ & $\mathbf{36.3}$ & $58.1$ & $71.6$ & $305$ & $127$\\
 $2000$ & $2$  & $\mathbf{1.88}$ & $2.76$ & $2.32$ & $52.6$ & $6.27$\\
 $2000$ & $4$  & $\mathbf{0.134}$ & $0.161$ & $\mathbf{0.134}$ & $7.97$ & ${0.696}$\\
\hline\hline
\end{tabular}
\caption{\label{table:compare}Average mean absolute deviation (over 100 repetitions) comparison between different methods. Other parameters used: $n=1000$ with $G=10$}
\end{center}
\end{table}

\subsection{Multiple change-points simulation}
\label{Sec:multi_sim}
The numerical studies so far have focused mainly on the single change-point estimation problem. In this subsection, we investigate the empirical performance of \texttt{groupInspect} in multiple change-point estimation tasks. We will compare its performance as implemented in Algorithm~\ref{alg:multiple change} to that of the \texttt{inspect} algorithms for estimating multiple change-points under different settings. We choose $n=1200$, $p\in\{500,1000\}$, $s\in\{3,10\}$, $G\in\{50,100\}$. Each data series contains three true change-points located at $300$, $600$ and $900$ with the $\ell_2$ norm of the change equal to  $\vartheta$, $1.5\vartheta$ and  $2\vartheta$ respectively. We vary  $\vartheta$ in $\{0.6,0.8,1,1.2,1.4\}$. For simplicity, we further assume that the same $s$ coordinates undergo change in all three change-points and that all groups have 10 elements.  We use the $\lambda$ tuning parameter choice suggested in Section~\ref{Sec:Tuning} for the \texttt{groupInspect} method and that suggested in \citet{WS2018} for the \texttt{inspect} algorithm. For the thresholding parameter $\xi$ of the wild binary segmentation recursion used in both \texttt{groupInspect} and \texttt{inspect}, we choose via Monte Carlo simulation. More precisely, we randomly generate 1000 data sets from the null model with no change-points and take the maximum absolute CUSUM statistics from Algorithm~\ref{alg:multiple change} and  \citet[Algorithm 4]{WS2018} as $\xi_g$ and $\xi_i$ respectively. We compare the performance of two algorithms using the Adjusted Rand index (ARI) of the estimated segmentation against the truth \citep{rand1971, HA1985}. 

From Figure~\ref{fig:ari}, we see that the \texttt{groupInspect} algorithm generally performs much better than the \texttt{inspect} algorithm in the multiple change-point localisation tasks. The advantage of \texttt{groupInspect} is more pronounced when the signal is sparser and when the dimension of the data is higher. 

\begin{figure}
\centering
\begin{tabular}{cc}
\includegraphics[width=0.45\textwidth]{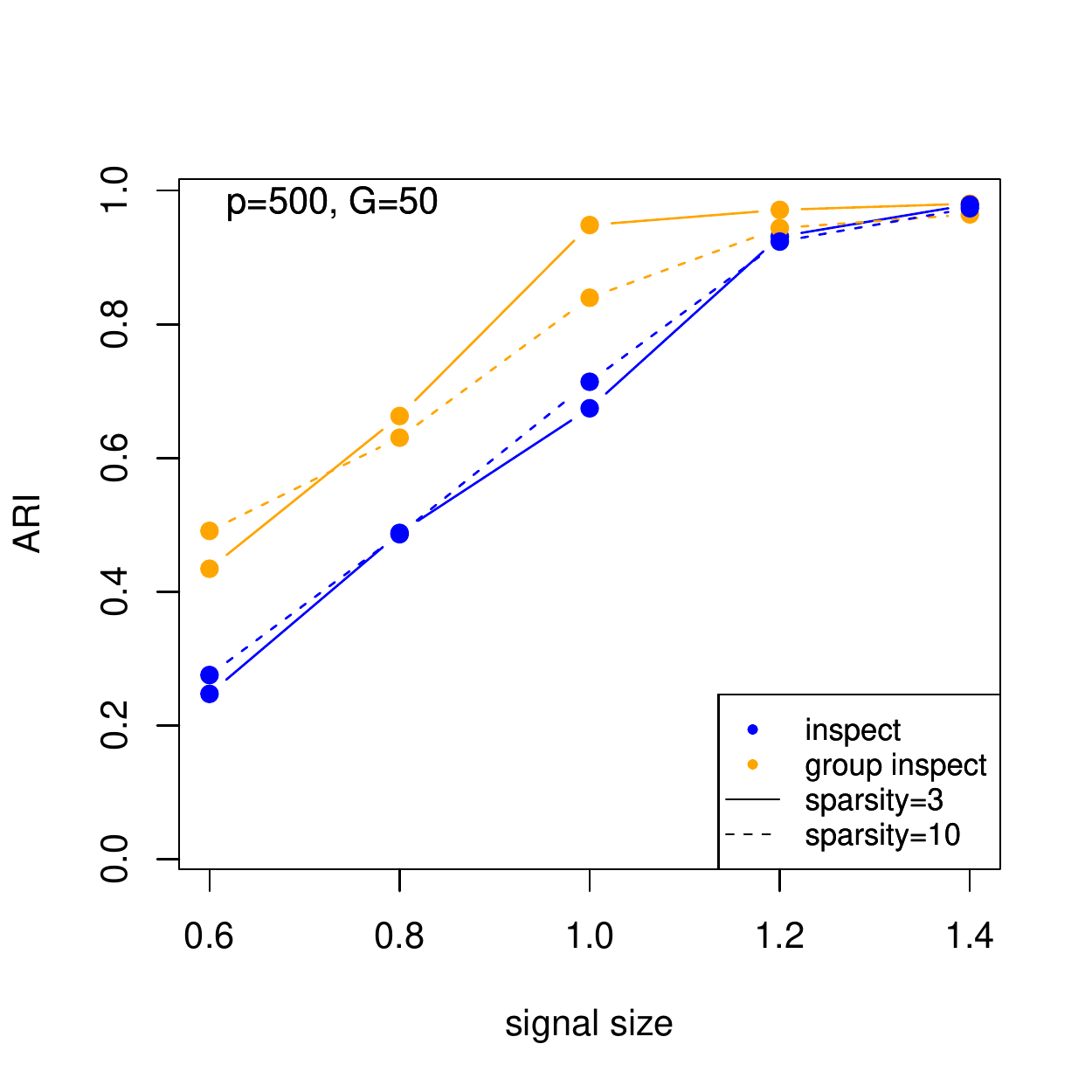} &
\includegraphics[width=0.45\textwidth]{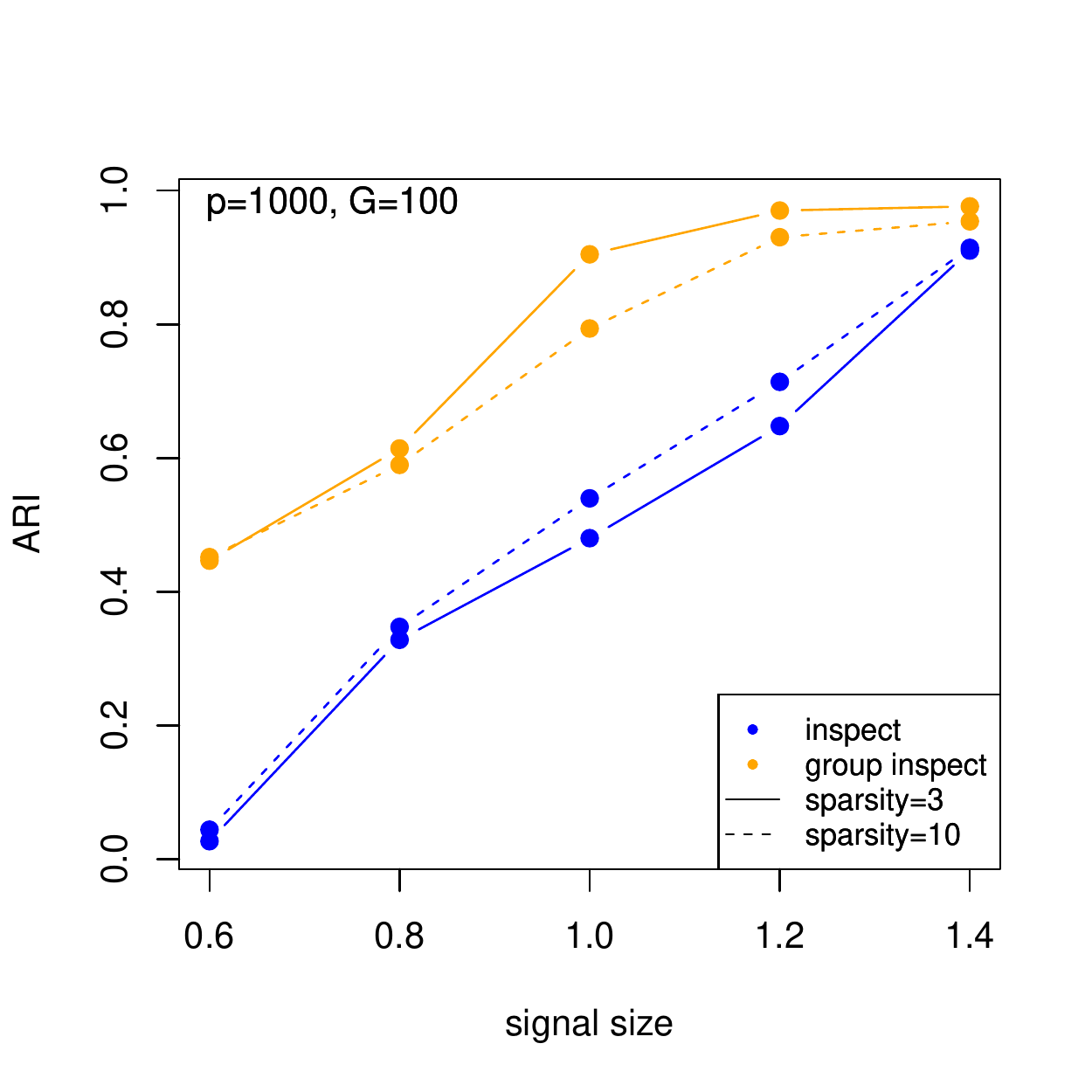}
\end{tabular}
\caption{\label{fig:ari}Average ARI comparsion between \texttt{groupInspect} and \texttt{inspect}. Left panel: $p=500,G=50$. Right panel: $p=1000,G=100$.}
\end{figure}

To further visualise the output of the two procedures, we plot the estimated change-point locations for one specific setting ($s=3$ and $\vartheta=1$) of each of the two panels in Figure~\ref{fig:ari}. The resulting histograms in Figure~\ref{fig:hist} shows that when $p=500$, \texttt{groupInspect} was better at picking out all three change-points with higher accuracies. When $p=1000$, \texttt{inspect} was only able to pick out the change at $t=600$ in most of the trials, whereas \texttt{groupInspect} was still able to identify even the weakest change signal at $t=300$ in a substantial fraction of all trials.

\begin{figure}[htbp]
\centering
\begin{tabular}{cc}
\includegraphics[width=0.45\textwidth]{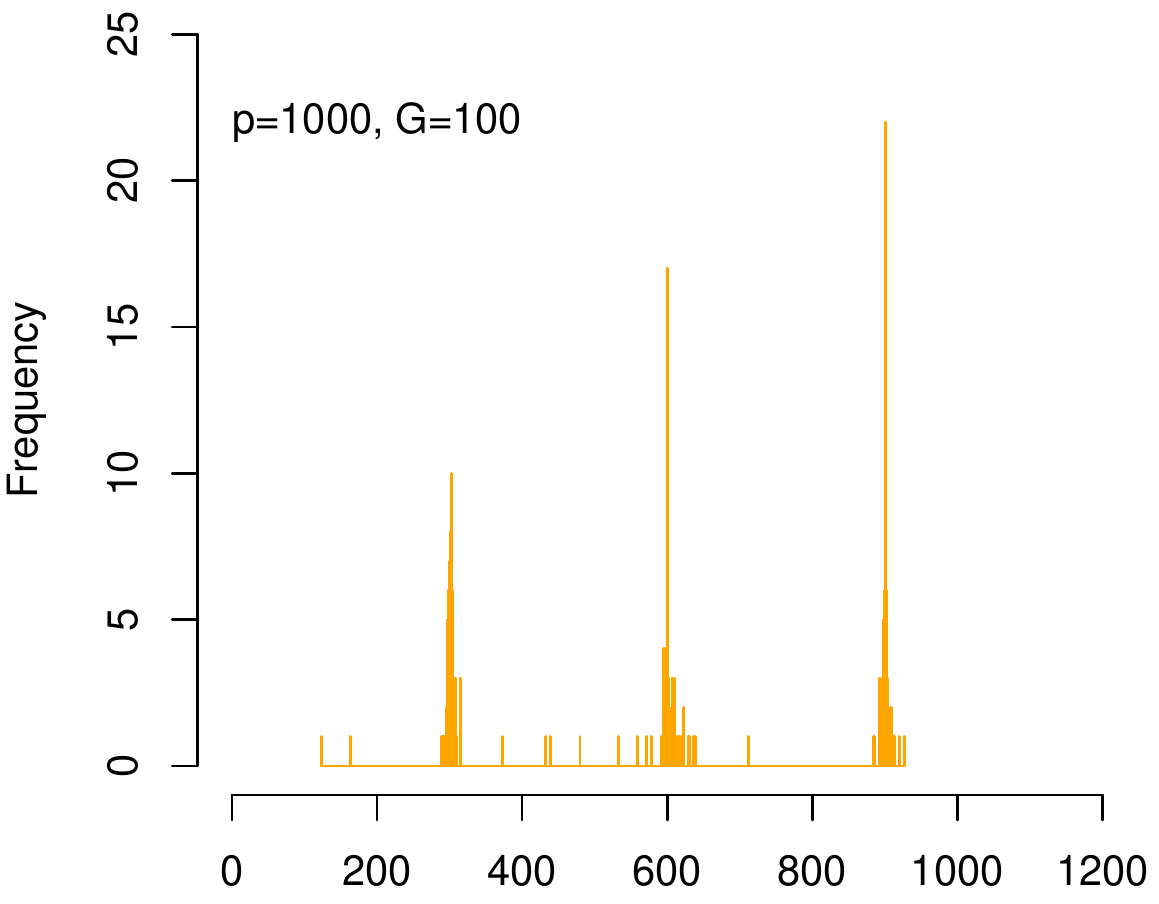} &
\includegraphics[width=0.45\textwidth]{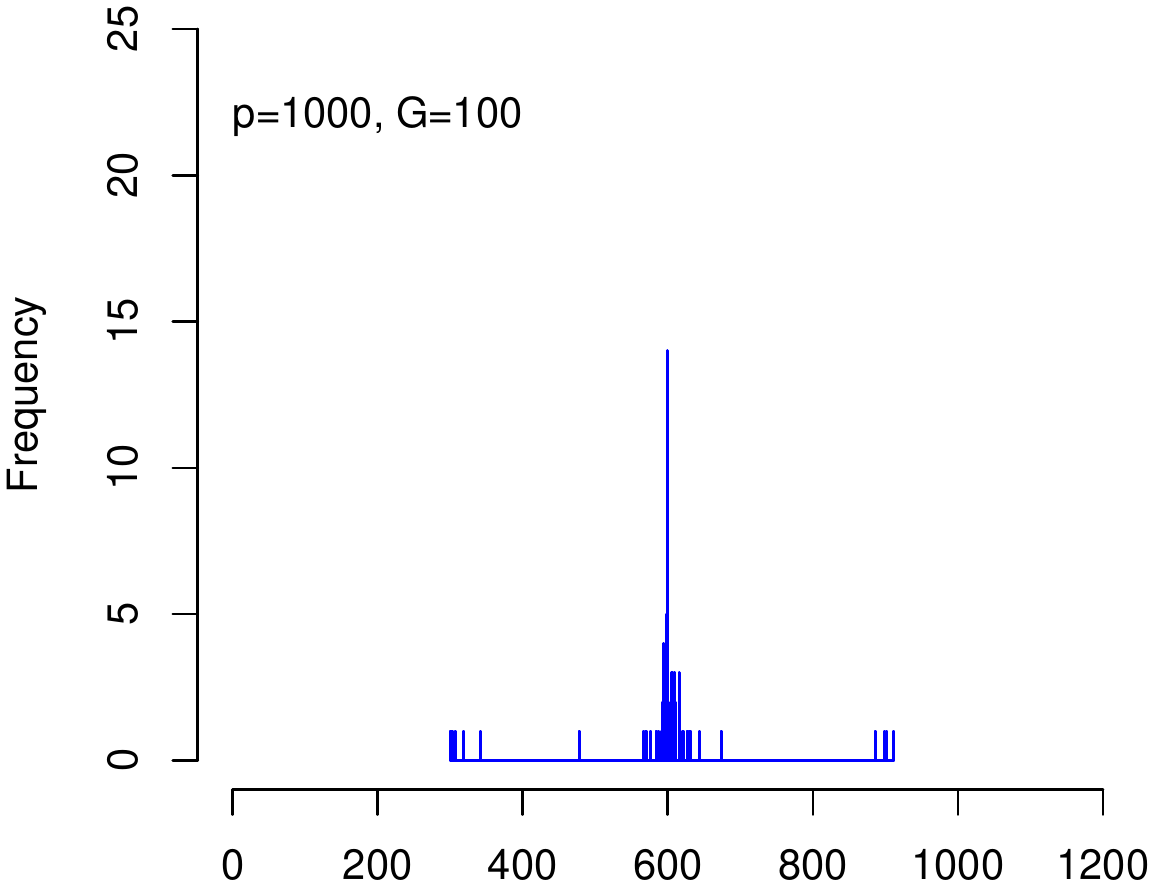}\\
\includegraphics[width=0.45\textwidth]{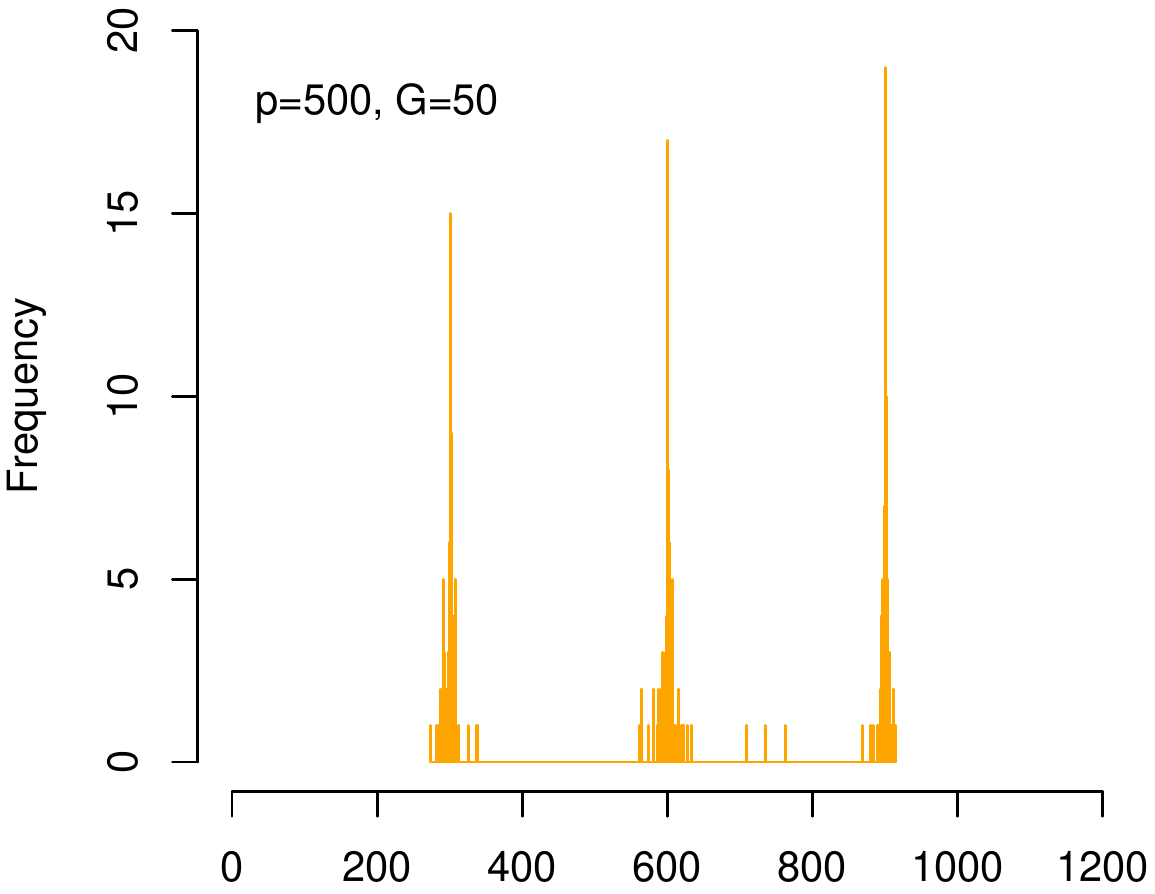} &
\includegraphics[width=0.45\textwidth]{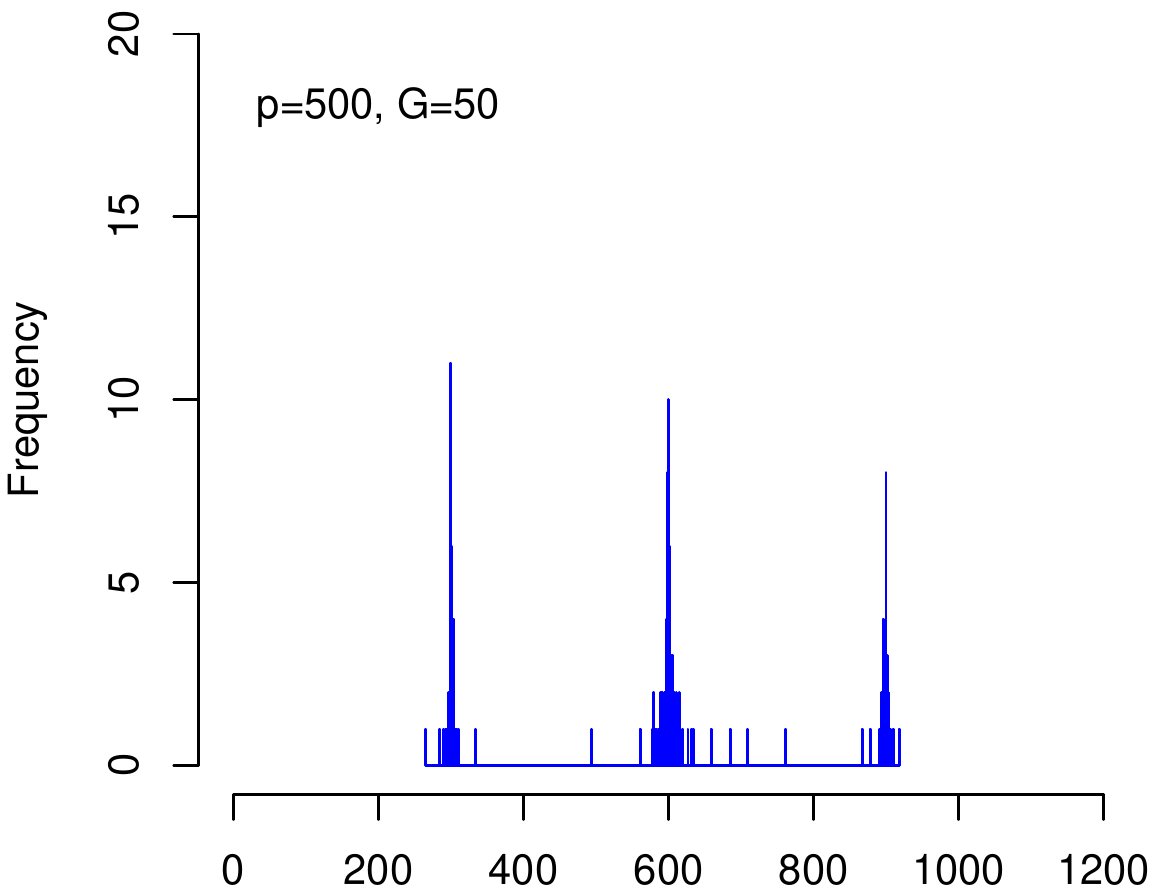}
\end{tabular}
\caption{\label{fig:hist}Histograms of estimated locations by \texttt{groupInspect} and \texttt{inspect} under two settings when $P=500,G=50$ and $p=1000, G=100$. Other parameter used: $s=3$, $\vartheta=1$ are fixed in both settings.  }
\end{figure}

\subsection{Real data analysis}
\label{sec:real}
In this section, we apply \texttt{groupInspect}  to a stock price data. The data consists of the logarithmic daily returns (computed from the adjusted closing prices) of S\&P 500 stocks during the period 1 January 2007 to 31 December 2011. Since not all companies remained in the S\&P 500 list and some companies have missing data at a few time points, we eventually selected 256 companies which have continuously traded throughout this this period to construct a multivariate time series of dimension $p=256$ and length $n=1259$.  We then divide the $256$ companies into $G=11$ non-overlapping groups according to respective Global Industry Classification Standard  sector memberships. We then rescale rows of the data matrix by their estimated standard deviation as in Section~\ref{Sec:Simulations}. We use the same procedure in Section~\ref{Sec:multi_sim} to choose thresholding parameter $\xi$. 

The \texttt{groupInspect} algorithm identifies the following change points $t=147$, $148$, $298$, $386$, $427$, $441$, $448$, $460$, $477$, $522$, $524$, $549$, $559$, $1158$, $1189$, as illustrated in Figure~\ref{fig:real}. We see a large number of changes being identified in the period between September and October 2008, which corresponds to the period when the financial crisis reaches a climax, and when the stock market is most volatile. 

 \begin{figure}[htbp]
\centering
\includegraphics[width=0.9\textwidth]{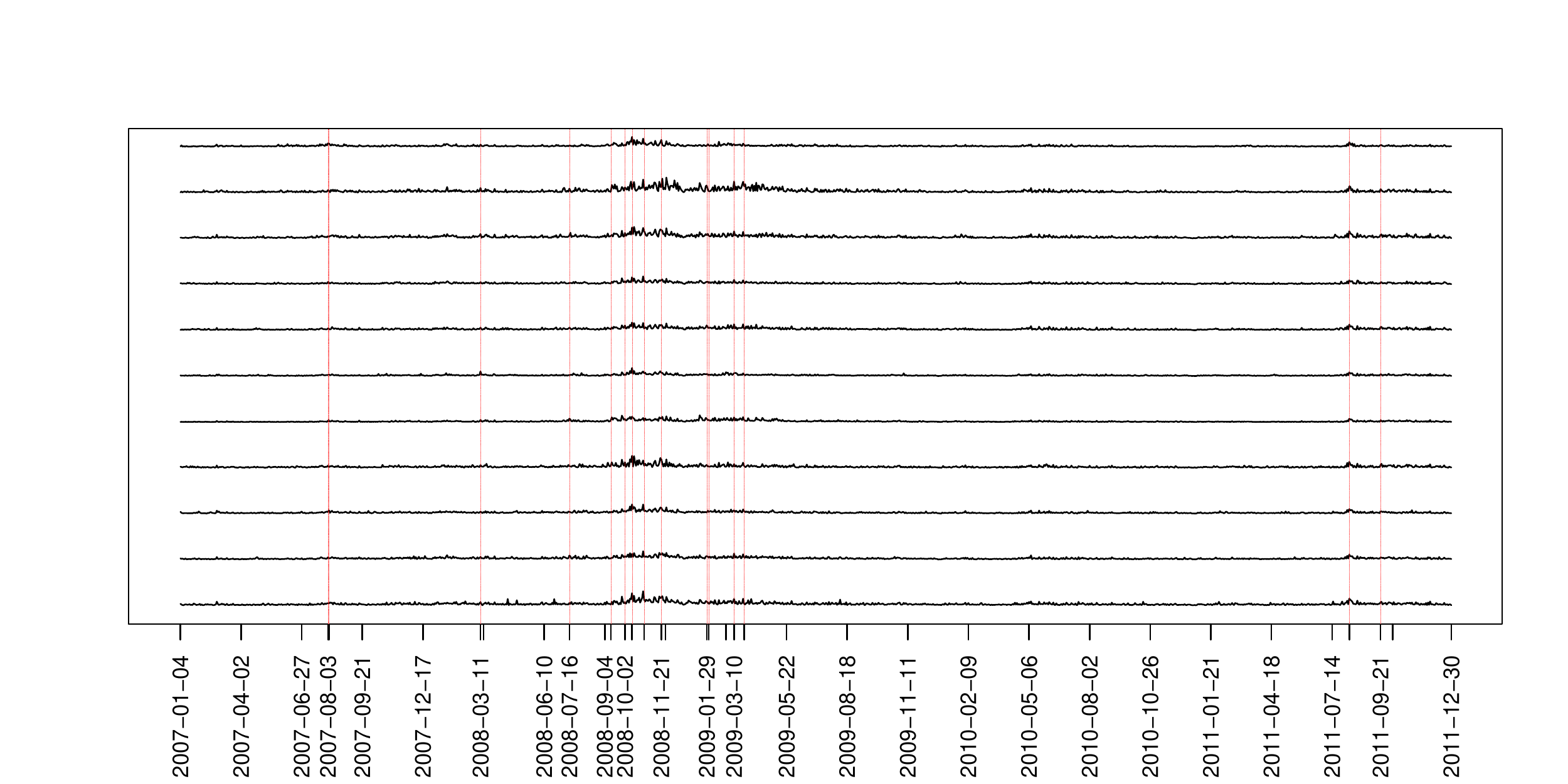} 
\caption{\label{fig:real}Estimated change point locations (red dashed lines) by \texttt{groupInspect} applied to the stock return data. For ease of illustration, we have plotted the $\ell_2$ norm of the returns of all stocks within each of the 11 groups over time.}
\end{figure}

\section{Proofs of main results}
\label{Sec:Proofs}
In this section, we will give the proof of our results in section~\ref{Sec:Theory}.
\subsection{Proof of Theorem~\ref{thm:upper bound}}
\begin{proof}
By \citet[equation (9)]{WS2018}, we can explicitly write the matrix $A = (A_{j,t})_{j\in[p], t\in[n-1]}$ by
\[
A_{j,t}=
\begin{cases}
  \sqrt{\frac{t}{n(n-t)}}(n-z)\theta_j& \text{if  $1\leq t\leq z$}, \\ \sqrt{\frac{n-t}{nt}}z\theta_j &\text{ if $z<t\leq n-1$}.
  \end{cases}
\]
Thus, we have 
\[
A=\theta \gamma^\top,
\]
where $\gamma =\frac{1}{\sqrt{n}}\bigr(\sqrt{\frac{1}{n-1}}(n-z),\sqrt{\frac{2}{n-2}}(n-z),\cdots \sqrt{z(n-z)},\sqrt{\frac{n-z-1}{z+1}}z,\cdots,\sqrt{\frac{1}{n-1}}z\bigr)^\top$. In particular, by By \citet[Lemma~3]{WS2018}, $A$ is a rank 1 matrix with $\|A\|_{\mathrm{op}} = \|\theta\|_2\|\gamma\|_2 \geq n\tau\vartheta/4$. By Proposition~\ref{Prop:SmallEvent} with $\delta=(nG)^{-2}$, we have 
\[
\mathbb{P}(\|T-A\|_{\mathrm{grp}*}>\lambda)<\frac{1}{Gn}.
\]
By Proposition~\ref{Prop:Deterministic}, on the event $\{\|T-A\|_{\mathrm{grp}_*}\leq \lambda\}$, we have
\[
\max\bigl\{\sin\angle(v,\hat{v}), \sin\angle(u,\hat{u})\bigr\}\leq\frac{32\lambda (C_1 k)^{1/2}}{n^{1/2}\tau\vartheta},
\]
as desired.
\end{proof}

\subsection{Proof of Theorem~\ref{thm:lower bound}}

\begin{proof}
We will use two different constructions to derive separate lower bounds of order $\sqrt{\sigma^2 s\log(G/s)/(n\tau\vartheta^2)}$  and  $\sqrt{\sigma^2 k/(n\tau\vartheta^2)}$ respectively. Without loss of generality, we may assume that $z<n/2$. 

For the first bound, let $s_0=s-1$, $G_0=G-1$, then in $\mathbb{R}^G$. By the Gilbert--Varshamov lemma as stated in \citet[Lemma~4.10]{Massart2007} (applied with $\alpha=3/4$ and $\beta=1/3$), we can construct a set $\mathcal{U}_0$ of $s_0$-sparse vectors in $\{0,1\}^{G_0}$, with cardinality at least $({G_0}/{s_0})^{s_0/5}$, such that the pairwise Hamming distance between any pair of vectors in $\mathcal{U}_0$ is at least ${s_0}/{2}$. Let $\epsilon \in (0,1)$ to be chosen later, we can define a set 
\[\mathcal{U} =\bigg\{\biggl(
\begin{array}{c}
\sqrt{1-\epsilon^2}\\
s_0^{-1/2}\epsilon u_0\\
\end{array}
\bigg): u_0\in \mathcal{U}_0\bigg\}\subseteq \mathbb{S}^{G-1}.
\]
We remark that for any pair of distinct $u, u'\in\mathcal{U}$, we have by construction that $\epsilon/\sqrt{2}\leq \|u'-u\|_2 \leq \epsilon$. 
We then define a map $\psi:\mathbb{R}^G\to\mathbb{R}^p$ such that for any $u\in\mathcal{U}$ and $j\in\mathcal{J}_g$, we have $\psi(u)_j = u_g p_g^{-1/2}$. Finally, let $\mathcal{V}=\{\psi(u):u\in\mathcal{U}\}$. We note that $\|\psi(u') - \psi(u)\|_2 = \|u'-u\|_2$. Therefore, for distinct $v,v'\in \mathcal{V}$, we have
\begin{equation}
\label{Eq:Fano1}
L(v',v)=\sqrt{1-(v^\top v')^2}=\frac{\|v'-v\|_2}{\sqrt 2}\geq\frac{\epsilon}{2}.
\end{equation}
Now, for each $v\in\mathcal{V}$, we can define a distribution $P_v \in \mathcal{P}_{n,p}(s, k, \tau, \vartheta,\sigma^2, (\mathcal{J}_g)_{g\in[G]})$, such that the pre-change mean is $-\vartheta v$ and the post-change mean is $0$ (we check that $P_v$ indeed satisfies the conditions of $ \mathcal{P}_{n,p}(s, k, \tau, \vartheta,\sigma^2, (\mathcal{J}_g)_{g\in[G]})$). Then for any distinct $v,v'\in\mathcal{V}$, we have 
\begin{equation}
\label{Eq:Fano2}
D(P_v \|P_{v'}) = z D(N_p(-v\vartheta,\sigma^2 I_p)\|N_p(-v'\vartheta,\sigma^2 I_p))= \frac{z\vartheta^2}{2}\|v- v'\|_2^2 \leq \frac{z\vartheta^2\epsilon^2}{2\sigma^2}.
\end{equation}
By~\eqref{Eq:Fano1} and~\eqref{Eq:Fano2}, we can apply Fano's lemma \citep[Lemma 3]{Yu1997} to obtain that
\begin{align*}
\inf_{\tilde v} \sup_{P\in\mathcal{P}_{n,p}(s,k,\tau,\vartheta,\sigma^2, (\mathcal{J}_g)_{g\in[G]})} \mathbb{E}_{P}L(\tilde v(X), v(P))&\geq
\inf_{\tilde v} \sup_{v\in\mathcal{V}} \mathbb{E}_{P_v}L(\tilde v(X), v)\\
&\geq
\frac{\epsilon}{4}\biggl\{1-\frac{z\vartheta^2 \epsilon^2/2\sigma^2+\log2}{(s_0/5)\log(G_0/s_0)}\biggr\}.
\end{align*}
By the condition $(s-1)\log(G/s)\geq 20$ in the theorem, we have $(s_0/5)\log(G_0/s_0) \geq 2\log 2$. Moreover, the choice of 
\[
\epsilon=\sqrt{\frac{\sigma^2 s_0\log (G_0/s_0)}{10z\vartheta^2}}
\]
ensures that $(s_0/5)\log(G_0/s_0)\geq 2z\vartheta^2\epsilon^2/\sigma^2$. Therefore,
\begin{equation}
\label{Eq:FirstLowerBound}
\inf_{\tilde v} \sup_{P\in\mathcal{P}_{n,p}(s,k,\tau,\vartheta,\sigma^2, (\mathcal{J}_g)_{g\in[G]})} \mathbb{E}_{P}L(\tilde v(X), v(P))\geq \frac{\epsilon}{16} \geq \frac{1}{72}\sqrt{\frac{\sigma^2 s\log (G/s)}{z\vartheta^2}}.
\end{equation}

For the second lower bound, let $g_1,\ldots,g_s$ be the indices of the $s$ groups with largest cardinalities. By the given condition of the Theorem, we have that  $\tilde k = \sum_{r=1}^{s} p_{g_r} = \sum_{r=1}^{s} p_{(G-r+1)} \geq k/2$.  Let $S=\cup_{r=1}^s \mathcal{J}_{g_r}$, so $|S| = \tilde k$. By \citet[Lemma 4.7]{Massart2007}, we can construct a subset $\mathcal{V}_0$ of $\{-1,1\}^{\tilde k_0}$ of cardinality at least $e^{\tilde k/8}$, such that any two points in the set are separated in Hamming distance by at least $\tilde k/4$. Construct
\[
\mathcal{V} = \biggl\{v: v_S = \begin{pmatrix}\sqrt{1-\epsilon^2}\\ \tilde k_0^{-1/2}\epsilon v_0\end{pmatrix} \text{ for some $v_0\in\mathcal{V}_0$ and $v_{S^\mathrm{c}} = 0$}\biggr\}.
\]
Therefore, for distinct $v,v'\in \mathcal{V}$, we have $\epsilon\leq\|v'-v\|_2\leq 2\epsilon$,then,
\[
L(v',v)=\sqrt{1-(v^\top v')^2}=\frac{\|v'-v\|_2}{\sqrt 2}\geq\frac{\epsilon}{\sqrt{2}}.
\]
Following the same derivation as in~\eqref{Eq:Fano2}, we have that
\[
D(P_v \|P_{v'}) = z D(N_p(-v\vartheta,\sigma^2 I_p)\|N_p(-v'\vartheta,\sigma^2 I_p))= \frac{z\vartheta^2}{2\sigma^2}\|v- v'\|_2^2 \leq 2z\vartheta^2\epsilon^2/\sigma^2.
\]
Again, we can use Fano's lemma \citep[Lemma 3]{Yu1997} to obtain that
\begin{align*}
\inf_{\tilde v} \sup_{v\in\mathcal{V}} \mathbb{E}_{P_v}L(\tilde v(X), v)\geq
\frac{\epsilon}{\sqrt{2}}\biggl\{1-\frac{2z\vartheta^2\epsilon^2/\sigma^2 +\log2}{\tilde{k}/8}\biggr\} \geq \frac{\epsilon}{\sqrt{2}}\biggl\{1-\frac{2z\vartheta^2\epsilon^2/\sigma^2 +\log2}{k/16}\biggr\}.
\end{align*}
Now, choose $\epsilon= \sigma k^{1/2}z^{-1/2}\vartheta^{-1}/4\sqrt{6}$. Since $k\geq 20$, we have $k/16\geq 9\log (2)/5$, so that
\begin{align}
\label{Eq:SecondLowerBound}
\inf_{\tilde v} \sup_{P\in\mathcal{P}_{n,p}(s,k,\tau,\vartheta,\sigma^2, (\mathcal{J}_g)_{g\in[G]})} \mathbb{E}_{P}L(\tilde v(X), v(P)) &\geq
\inf_{\tilde v} \sup_{v\in\mathcal{V}} \mathbb{E}_{P_v}L(\tilde v(X), v)\nonumber\\
&\geq\frac{\epsilon}{9\sqrt{2}}\geq\frac{1}{72\sqrt{3}}\sqrt{\frac{\sigma^2 k}{z\theta^2}}.
\end{align}
The desired result follows by combining~\eqref{Eq:FirstLowerBound} with~\eqref{Eq:SecondLowerBound}, and noting that $z\geq n\tau$. 
 \end{proof}
 
\subsection{Proof of Theorem~\ref{Thm:change-pointEstimate}}
\begin{proof}
Recall the definition of $X^{(2)}$ and let $T^{(2)}=\mathcal{T}(X^{(2)})$. Define similarly $\mathbf{\mu}^{(2)}=(\mu_1^{(2)},\ldots,\mu_{n_1}^{(2)})\in\mathbb{R}^{p\times n_1}$ and a random $W^{(2)}=(W_1^{(2)},\ldots,W_{n_1}^{(2)})$ taking values in $\mathbb{R}^{p\times n_1}$ by $\mu_t^{(2)}=\mu_{2t}$ and $W_t^{(2)}=W_{2t}$. Now, let $A^{(2)}=\mathcal{T}(\mathbf{\mu}^{(2)})$ and $E^{(2)}=\mathcal{T}(W^{(2)})$. We also write $\bar{X}=(\hat{v}^{(1)})^\top X^{(2)}$,$\bar{\mu}=(\hat{v}^{(1)})^\top\mathbf{\mu}^{(2)}$,$\bar{W}=(\hat{v}^{(1)})^\top W^{(2)}$,$\bar{A}=(\hat{v}^{(1)})^\top A^{(2)}$, $\bar{E}=(\hat{v}^{(1)})^\top E^{(2)}$ and $\bar T=(\hat{v}^{(1)})^\top T^{(2)}$ for the one-dimensional projected images. Note that by linearity, we have $\bar{T}=\mathcal{T}(\bar{X})$, $\bar A=\mathcal{T}(\bar{\mu})$ and $\bar E=\mathcal{T}(\bar{W})$,

Now, conditional on $\hat{v}^{(1)}$, the random variables $\bar{X}_1,\ldots,\bar{X}_{n_1}$ are independent with 
\[
\bar{X}_t \mid \hat{v}^{(1)}\sim N(\bar{\mu}_t,\sigma^2)
\]
and the row vector $\bar{\mu}$ undergoes a single change at $z^{(2)} =z/2$ with magnitude of change 
\[
\bar{\theta}=\bar{\mu}_{z^{(2)}+1}-\bar{\mu}_{z^{(2)}}=\hat{v}^{(1)\top}\theta.
\]
Finally, let $\hat{z}^{(2)}\in\argmax_{1\leq t\leq n_1-1}|\bar{T_t}|$, so the first component of the output of the algorithm is $\hat{z}=2\hat{z}^{(2)}$. Consider the set 
\[
\Upsilon=\{u \in\mathbb{S}^{p-1}:\angle(u,v)\leq \pi/6\}.
\]
By Condition \eqref{eq:cond 1} and Theorem~\ref{thm:upper bound}, we have that
\begin{equation}
\label{Eq:Prob1}
\mathbb{P}(\hat{v}^{(1)} \in \Upsilon)\geq 1-\frac{1}{n_1 G}.
\end{equation}
Moreover, on the event $\{\hat v^{(1)} \in \Upsilon\}$, we have that $\bar{\theta}\geq \sqrt{3}\vartheta/2$. Set $\lambda_1=\sigma(1+\sqrt{4 \log n})$, we have by Proposition~\ref{Prop:SmallEvent} that 
\begin{equation}
\label{Eq:Prob2}
    \mathbb{P}(\|\bar{E}\|_{\infty}\geq\lambda_1)=\mathbb{P}(\|\bar{E}\|_{\mathrm{grp}*}\geq \lambda)\leq \frac{1}{n_1}.
\end{equation}
Since $\bar{T}=\bar{A}+\bar{E}$ and $(\bar{A_t})_t$ and  $(\bar{T_t})_t$ are respectively maximized at $t=z^{(2)}$ and $t=\hat{z}^{(2)}$, we have on the event $\Omega_0=\{\hat{v}_1\in \Upsilon,\|\bar{E}\|_{\infty}\geq\lambda_1\}$ that
\begin{align*}
\bar{A}_{z^{(2)}}-\bar{A}_{\hat{z}^{(2)}}&=(\bar{A}_{z^{(2)}}-\bar{T}_{\hat{z}^{(2)}})+(\bar{T}_{z^{(2)}}-\bar{T}_{\hat{z}^{(2)}})+(\bar{T}_{\hat{z}^{(2)}}-\bar{A}_{\hat{z}^{(2)}})\\
&\leq |\bar{A}_{z^{(2)}}-\bar{A}_{\hat{z}^{(2)}}|+|\bar{T}_{\hat{z}^{(2)}}-\bar{A}_{\hat{z}^{(2)}}|\leq 2\lambda_1.
\end{align*}
Hence, by \citet[Lemma 7 in the online supplement]{WS2018}, on the event $\Omega_0$, we have that 
\begin{equation*}
\label{eq:cond 3}
    \frac{|\hat{z}^{(2)}-z^{(2)}|}{n_1\tau}\leq\frac{3\sqrt{6}\lambda_1}{\bar{\theta}(n_1\tau)^{1/2}}\leq\frac{6\sqrt{2}\sigma(1+\sqrt{4\log n_1})}{\vartheta\sqrt{n\tau}}.
\end{equation*}
Now we define the event
\[
\Omega_1=\biggl\{\biggl|\sum_{r=1}^s\bar{W_t}-\sum_{r=1}^t\bar{W}\biggr|\leq\lambda_1 \sqrt{|s-t|},\quad for\text{ }all\text{ } 0\leq t\leq n_1, s\in\{0,z^{(2)},n_1\}\biggr\}.
\]
By \citet[Lemma~5]{WS2018}, we have that 
\begin{equation}
\label{Eq:Prob3}
\mathbb{P}(\Omega_1^{\mathrm{c}}) \leq 4e^{-\lambda_1^2/4}\{2\log n_1+\log z^{(2)} + \log(n_1-z^{(2)}\} \leq 16\log n e^{-\lambda_1^2/4} \leq \frac{16\log n}{n}.
\end{equation}
Following the proof of Theorem~1 of \citet{WS2018}, we have on $\Omega_0\cap\Omega_1$ that 
\[
 1 \leq\frac{6\sqrt{3}\lambda_1}{\bar{\theta}|\hat{z}^{(2)}-z^{(2)}|^{1/2}}+\frac{12\sqrt{6}\lambda_1}{\bar{\theta}(n_1\tau)^{1/2}}\\
 \leq\frac{12\sqrt{2}\sigma(1+\sqrt{4\log n})}{\vartheta \sqrt{|z-\hat{z}}|}+\frac{48\sigma(1+\sqrt{4\log n})}{\vartheta\sqrt{n\tau}}
\]
From \eqref{eq:cond 1} for $C\geq 96$, we have on $\Omega_0\cap\Omega_1$ that
\[
|\hat{z}-z|\leq C'\vartheta^{-2}\sigma^2(1+\sqrt{4\log n}).
\]
Finally, by~\eqref{Eq:Prob1},~\eqref{Eq:Prob2} and~\eqref{Eq:Prob3} we have that
\[
\mathbb{P}(\Omega_0\cap\Omega_1) \geq 1 - \frac{1}{n_1G} - \frac{1}{n_1} - \frac{16\log n}{n} \geq 1 - \frac{20\log n}{n},
\]
as desired.
\end{proof}

\section{Ancillary results}
\label{Sec:Ancillary}
We collect in this section all ancillary propositions and lemmas used in the paper. For all results in this section, we assume that we are given a grouping $(\mathcal{J}_g)_{g\in [G]}$ of $[p]$ and the associated group norm $\|\cdot\|_{\mathrm{grp}}$. It is useful to define the following counterpart to the group norm. For any $R \in \mathbb{R}^{p\times n}$ and a grouping $(\mathcal{J}_g)_{g\in[G]}$ of $[p]$, we define 
\begin{equation}
\label{Eq:GrpDualNorm}
	\left\lVert R \right\rVert_{\mathrm{grp*}}=p_g^{-1/2}\max_{g\in[G]}\max_{t\in[n]} \|R_{\mathcal{J}_g,t}\|_2.
\end{equation}
\begin{lemma}
\label{Lemma:DualNorm}
The norm $\|\cdot\|_{\mathrm{grp*}}$ is a dual to $\|\cdot\|_{\mathrm{grp}}$ with respect to the inner product $\langle \cdot,\cdot\rangle$ on $\mathbb{R}^{p\times n}$.
\end{lemma}
\begin{proof}
To prove the lemma, it suffices to show that $\left\lVert M \right\rVert_{\mathrm{grp}}=\sup\limits_{\left\lVert R \right\rVert_{\mathrm{grp*}}\leq 1}\langle R,M \rangle$ for all $M \in \mathbb{R}^{p\times (n-1)}$. First, for any $M \in \mathbb{R}^{p\times (n-1)}$, let $M_{\mathcal{J}_g,t}$ be the $t$th column of $M_{\mathcal{J}_g}$. Define $\tilde{R} = \tilde R(M)$ such that
\[
\tilde R_{\mathcal{J}_g,t} = \frac{p_g^{1/2} M_{\mathcal{J}_g,t}}{\max\bigl\{\|M_{\mathcal{J}_g,t}\|_2, 1\bigr\}}.
\]
Then, $\|\tilde R\|_{\mathrm{grp*}}\leq\max_{g\in[G]}\max_{t\in[n-1]}p_g^{-1/2}p_g^{1/2}\frac{\|M_{\mathcal{J}_g,t}\|_2}{\|M_{\mathcal{J}_g,t}\|_2}=1$. Hence, \begin{align*}
\sup_{\left\lVert R \right\rVert_{\mathrm{grp*}}\leq 1}\langle R,M\rangle \geq \langle \tilde{R},M\rangle&=\sum_{g=1}^{G}\sum_{t=1}^{n-1}p_g^{1/2}\frac{\langle M_{\mathcal{J}_g,t}, M_{\mathcal{J}_g,t}\rangle }{\|M_{\mathcal{J}_g,t}\|_{2}}\\
&=\sum_{g=1}^{G}\sum_{t=1}^{n-1}p_g^{1/2}\|M_{\mathcal{J}_g,t}\|_2=\|M\|_{\mathrm{grp}}.
\end{align*}
On the other hand, for any $R$ such that $\|R\|_{\mathrm{grp*}}\leq 1$, we have $\|R_{\mathcal{J}_g,t}\|_2\leq p_g^{1/2} $ for all $g$ and $t$. Consequently, by the Cauchy--Schwarz inequality,
\begin{align*}
\langle R,M\rangle &= \sum_{g\in[G]}\sum_{t\in[n-1]} \langle R_{\mathcal{J}_g,t},M_{\mathcal{J}_g,t}\rangle \leq \sum_{g\in[G]}\sum_{t\in[n-1]}\|R_{\mathcal{J}_g,t}\|_2\|M_{\mathcal{J}_g,t}\|_2 \\
&\leq \sum_{g\in[G]}\sum_{t\in[n-1]}p_g^{1/2}\|M_{\mathcal{J}_g,t}\|_2 = \|M\|_{\mathrm{grp}},
\end{align*}
thus establishing the result.
\end{proof}

\begin{prop}
\label{prop: closed form}
Let $\mathcal{S}=\{M\in \mathbb{R}^{p\times(n-1)}:\left\lVert M\right\rVert_{\mathrm{F}}\leq 1\}$. For $T \in \mathbb{R}^{p\times(n-1)}$, $\lambda>0$, we have
\[
\argmax_{M\in \mathcal{S}}\Big\{\langle T,M\rangle-\lambda\|M\|_{\mathrm{{grp}}}\Big\}  =\frac{T-R^*}{\left\lVert T-R^*\right\rVert_{\mathrm{F}}},
\]
where $R^*$ satisfies $R_{\mathcal{J}_g,t}^*=T_{\mathcal{J}_g,t} \min\big\{\frac{\lambda p_g^{1/2}}{\|T_{\mathcal{J}_g,t}\|_{\mathrm{F}}}, 1\big\}$.

\end{prop}

\begin{proof}
Define functions $h:\mathbb{R}^{p\times (n-1)} \times \mathbb{R}^{p\times (n-1)} \to \mathbb{R}$ and $f,g:\mathbb{R}^{p\times (n-1)} \to \mathbb{R}$ such that for $M, R \in \mathbb{R}^{p\times (n-1)}$, $h(M,R)=\langle T-\lambda R,M\rangle$ and $f(M)= \inf_{\|R\|_{\mathrm{grp*}}\leq 1} h(M,R)$ and  $g(R) = \sup_{M\in \mathcal{S}} h(M,R)$. By~\eqref{Eq:GrpDualNorm} and Lemma~\ref{Lemma:DualNorm}, we have that 
\[
\langle T,M\rangle-\lambda \|M\|_{\mathrm{grp}}=\langle T,M\rangle-\lambda \sup_{\|R\|_{\mathrm{grp*}}\leq 1}\langle R,M \rangle=\inf\limits_{\left\lVert R \right\rVert_{\mathrm{grp*}}\leq 1} \langle T-\lambda R, M\rangle = f(M).
\]
By the minimax equality theorem \citep[Theorem~1]{Fan1953}, we obtain that
\[
\sup_{M\in \mathcal{S}}f(M)=\sup_{M\in \mathcal{S}}\inf_{\|R\|_{\mathrm{grp*}}\leq 1}h(M,R)=\inf_{\|R\|_{\mathrm{grp*}}\leq 1}\sup_{M\in \mathcal{S}}h(M,R)=\inf_{\|R\|_{\mathrm{grp*}}\leq 1}g(R).
\]
Observe that $g(R) = \|T-\lambda R\|_{\mathrm{F}}$. To find the optimiser $R^* \in \argmin_{\|R\|_{\mathrm{grp*}}\leq 1}\|T-\lambda R\|_{\mathrm{F}}$, we consider the $G$ groups individually. For each group $g$, and in the $t$th column, if $\|T_{\mathcal{J}_g,t}\|_{2} \leq \lambda p_g^{1/2}$, then $R^*_{\mathcal{J}_g,t} = T_{\mathcal{J}_g,t}/\lambda$; and if $\|T_{\mathcal{J}_g,t}\|_{2} > \lambda p_g^{1/2}$, then $R^*_{\mathcal{J}_g,t} = p_g^{1/2}T_{\mathcal{J}_g,t} / \|T_{\mathcal{J}_g,t}\|_{2}$. Since the minimizer of $g(R)$ is unique, we have that

\[
\argmax_{M\in \mathcal{S}} f(M) = \argmax_{M\in \mathcal{S}} h(M, R^*) = \frac{T-\lambda R^*}{\|T-\lambda R^*\|_{\mathrm{F}}},
\]
as desired.
\end{proof}

\begin{lemma}
\label{Lemma:GroupHolder}
For any $A, B \in \mathbb{R}^{p\times n}$, we have $\langle A,B \rangle \leq \|A\|_{\mathrm{grp}}\|B\|_{\mathrm{grp*}}$.
\end{lemma}

\begin{proof}
By Cauchy--Schwarz inequality, we have that 
\begin{align*}
\langle A,B\rangle &=\sum_{g,t}\langle A_{\mathcal{J}_g,t},B_{\mathcal{J}_g,t}\rangle\leq\sum_{g\in[G],t\in[n]}\|A_{\mathcal{J}_g,t}\|_{\mathrm{F}}\|B_{\mathcal{J}_g,t}\|_{\mathrm{F}}\\
&\leq \biggl(\sum_{g\in[G],t\in[n]} p_g^{1/2}\|A_{\mathcal{J}_g,t}\|_{\mathrm{F}} \biggr)\biggl(\max_{g\in[G],t\in[n]} p_g^{-1/2}\|B_{\mathcal{J}_g,t}\|_{\mathrm{F}} \biggr) = \|A\|_{\mathrm{grp}}\|B\|_{\mathrm{grp*}}.
\end{align*}
as desired.
\end{proof}

\begin{lemma}
\label{Lem:GrpNormToFrobNorm}
Let $p_g = |\mathcal{J}_g|$ and suppose further that there exists a universal constant $C_1>0$, such that $\max_{j\in[p]} |\{g:j\in \mathcal{J}_g\}| \leq C_1$. Then, for any $M\in\mathbb{R}^{p\times n}$, we have  $\|M\|_{\mathrm{grp}}\leq (C_1 n\sum_g p_g)^{1/2}\|M\|_{\mathrm{F}}$.
\end{lemma}

\begin{proof}
Define $m$ with $m_{\mathcal{J}_g,t}=\|M_{\mathcal{J}_g,t}\|_{\mathrm{F}}$. Then by applying the Cauchy--Schwarz inequality twice, we have
\begin{align*}
   \|M\|_{\mathrm{grp}}&=\sum_{g\in[G]}p_g^{1/2} \sum_{t=1}^n \|M_{\mathcal{J}_g,t}\|_2 \leq\sum_{g\in [G]}(n p_g)^{1/2} \|M_{\mathcal{J}_g}\|_{\mathrm{F}}\\
   &\leq\sqrt{n}\biggl(\sum_{g\in[G]}p_g\biggr)^{1/2}\biggl(\sum_{g\in[G]}\|M_{\mathcal{J}_g}\|_{\mathrm{F}}^2\biggr)^{1/2} \leq \biggl(C_1 n\sum_{g\in [G]} p_g\biggr)^{1/2}\|M\|_{\mathrm{F}},
\end{align*}
as desired.
\end{proof}

\begin{prop}
\label{Prop:Deterministic}
Let $p_g = |\mathcal{J}_g|$ and suppose further that there exists a universal constant $C_1>0$, such that $\max_{j\in[p]} |\{g:j\in \mathcal{J}_g\}| \leq C_1$. Let $A$ be a rank one matrix with $A = \delta vu^\top$ for $\delta > 0$, $\|v\|_2 = \|u\|_2 = 1$ and $\sum_{g: v_{\mathcal{J}_g} \neq 0} p_{g}\leq k$. Suppose $T\in\mathbb{R}^{p\times (n-1)}$ satisfies $\|T-A\|_{\mathrm{grp*}}\leq \lambda$ for some $\lambda>0$, and let $\mathcal{S} = \{M\in \mathbb{R}^{p\times (n-1)}: \|M\|_{\mathrm{F}} \leq 1\}$. Then, for any 
\[
\hat M\in \argmax_{M\in \mathcal{S}}\big\{\langle T,M\rangle- \lambda \|M\|_{\mathrm{grp}}\bigr\},
\]
we have
\[
\|vu^\top -\hat{M}\|_{\mathrm{F}}\leq\frac{4\lambda (C_1 nk)^{1/2}}{\delta},
\]
and
\[
\max\{\sin\angle(v,\hat{v}), \sin\angle(u,\hat{u})\}\leq\frac{8\lambda (C_1 nk)^{1/2}}{\delta}.
\]
\end{prop}
\begin{proof}
Define $\mathcal{G}_0=\{g: v_{\mathcal{J}_g}\neq 0\}$. Since $vu^\top \in \mathcal{S}$, from the basic inequality, we have
\begin{equation}
\label{Eq:BasicIneq}
\langle T, vu^\top\rangle - \lambda\|vu^\top\|_{\mathrm{grp}} \leq \langle T, \hat M \rangle - \lambda\|\hat M\|_{\mathrm{grp}}.
\end{equation}
When $\|A-T\|_{\mathrm{grp*}}\leq\lambda$,  or equivalently, $p_g^{-1/2}\|A_{\mathcal{J}_g, t}-T_{\mathcal{J}_g, t}\|_2 \leq\lambda$ for all $g\in[G]$ and $t\in [n-1]$, we have by \citet[Lemma~2]{WS2018} and~\eqref{Eq:BasicIneq} that
\begin{align*}
\|vu^\top -\hat{M}\|_{\mathrm{F}}^2&\leq \frac{2}{\delta}\langle A, vu^\top - \hat{M}\rangle \leq \frac{2}{\delta}\bigl(\langle T, vu^\top - \hat{M}\rangle + \langle A-T, vu^\top - \hat{M}\rangle\bigr)\\
&\leq \frac{2\lambda}{\delta}\bigl(\|vu^\top\|_{\mathrm{grp}}-\|\hat{M}\|_{\mathrm{grp}}+\|vu^\top -\hat{M}\|_{\mathrm{grp}}\bigr)\\
& = \frac{4\lambda}{\delta}\sum_{g\in \mathcal{G}_0}\sum_{t\in[n-1]} \|(vu^\top - \hat M)_{\mathcal{J}_g, t}\|_2 \leq\frac{4\lambda (C_1 nk)^{1/2}}{\delta}\|vu^\top -\hat{M}\|_{\mathrm{F}},
\end{align*}
where we used Lemma~\ref{Lemma:GroupHolder} in the penultimate inequality and Lemma~\ref{Lem:GrpNormToFrobNorm} in the final bound. This proves the first claim of the proposition, and the second claim follows from the first by the same argument as used in \citet[online supplement (18) and (19)]{WS2018}.
\end{proof}

\begin{prop} 
\label{Prop:SmallEvent}
Let $W$ be an $p\times n$ random matrix with independent $N(0,\sigma^2)$ entries and set $E=\mathcal{T}(W)$. Let $p_g = |\mathcal{J}_g|$ with $p_*=\min_{g\in[G]} p_g$. For any $\delta \in (0,1)$ and $\lambda=\sigma(1+\sqrt{2p_*^{-1} \log(1/\delta)})$, we have that 
\[
\mathbb{P}(\|E\|_{\mathrm{grp*}}>\lambda)\leq (n-1)G\delta.
\]
\end{prop}
\begin{proof}
By the definition of the CUSUM transformation $\mathcal{T}$ in~\eqref{eq:cusum}, we have that $E_{\mathcal{J}_g,t}\sim N(0, \sigma^2 I_{p_g})$, and $\|E_{\mathcal{J}_g,t}/\sigma \|_2^2\sim \chi_{p_g}^2$ for every $g\in[G]$ and $t\in[n-1]$. Consequently, by a union bound, we have 
\begin{align*}
\mathbb{P}(\|E\|_{\mathrm{grp*}}>\lambda)&\leq \sum_{g\in[G]}\sum_{t\in[n-1]}\mathbb{P}(\|E_{\mathcal{J}_g,t}\|_2^2>p_g\lambda^2)\\
&=\sum_{g\in[G]}\sum_{t\in[n-1]}\mathbb{P}\biggl\{\frac{\|E_{\mathcal{J}_g,t}\|_2^2}{\sigma^2 p_g}> \Bigl(1+\sqrt{2 p_*^{-1}\log(1/\delta)}\Bigr)^2\biggr\}\\
&\leq \sum_{g\in[G]}\sum_{t\in[n-1]}\mathbb{P}\biggl\{\frac{\|E_{\mathcal{J}_g,t}\|_2^2}{\sigma^2 p_g}> 1+2\sqrt{\frac{\log(1/\delta)}{p_g}}+\frac{2(\log{1/\delta})}{p_g} \biggr\} \\
& \leq (n-1)G\delta,
\end{align*}
as desired, where we used \citet[Lemma~1]{LaurentMassart2000} in the final inequality. 
\end{proof}

\end{document}